\documentclass[11pt,a4paper]{article}

\usepackage{graphicx,subfigure}
\usepackage{amsmath,amssymb,amsthm,xspace}
\usepackage[usenames]{color}
\usepackage{tikz}
\usepackage[square, comma, sort&compress]{natbib}
\setcitestyle{numbers}
\usepackage{enumerate}  
\usepackage{ifthen}
\usepackage[margin=1in]{geometry}
\usepackage{algorithm,algorithmic}
\newtheorem{thm}{Theorem}

\newtheorem{cor}[thm]{Corollary}
\newtheorem{clm}{Claim}
\newtheorem{lem}[thm]{Lemma}

\newtheorem{Theorem*}{Theorem}

\newtheorem{Claim*}{Claim}

\newtheorem{CounterExample*}{$\overline{\hbox{\bf Example}}$}

\newtheorem{Example}{Example}
\newtheorem{Example*}{Example}

\newtheorem{Intuition*}{Intuition}
\newtheorem{Joke*}{Joke}

\newtheorem{Lemma*}{Lemma}
\newtheorem{Open problem}{Open problem}

\newtheorem{Question*}{Question}
\newtheorem{Remark*}{Remark}

\def \bSubexa    {\begin{subexa}}

\def \Proof    {\skpbld{Proof}}

\newcommand{\ignore}[1]{{}}

\newcommand{\EE}{\mathbb{E}}
\newcommand{\CC}{\mathbb{C}}
\newcommand{\QQ}{\mathbb{Q}} 
\newcommand{\ZZ}{\mathbb{Z}} 

\def \cS     {{\cal S}}
\def \cT     {{\cal T}}

\newcommand{\eg}{\textit{e.g.,}\xspace}

\newcommand{\ie}{\textit{i.e.,}\xspace}  

\def \mynote#1{{}}

\def \eqed    {\eqno{\qed}}

\def \upto  {{,}\ldots{,}}

\def \setzo    {\{0{,}1\}}

\def \setzon   {\setzo^n}

\def \sets#1{{\{#1\}}}

\def \union  {\cup}

\def \ceil#1{{\lceil{#1}\rceil}}

\def \floor#1{{\lfloor{#1}\rfloor}}
\def \Floor#1{{\left\lfloor{#1}\right\rfloor}}

\def \paren#1{{({#1})}}
\def \Paren#1{{\left({#1}\right)}}

\def \gcd#1#2{{{\rm gcd}\paren{{#1},{#2}}}}

\newcommand{\ed}{\stackrel{\mathrm{def}}{=}}

\def\ignore#1{}

\renewcommand{\Proof}{\par\emph{Proof:}\ }
\newcommand{\prfarg}[1]{\par\emph{Proof of\ {#1}:}\ }

\newcommand{\bi}{\begin{itemize}}
\newcommand{\ei}{\end{itemize}}

\def\orpro{\mathop{\mathchoice
   {\vee\kern-.49em\raise.7ex\hbox{$\cdot$}\kern.4em}
   {\vee\kern-.45em\raise.63ex\hbox{$\cdot$}\kern.2em}
   {\vee\kern-.4em\raise.3ex\hbox{$\cdot$}\kern.1em}
   {\vee\kern-.35em\raise2.2ex\hbox{$\cdot$}\kern.1em}}\limits}

\def\andpro{\mathop{\mathchoice
 {\wedge\kern-.46em\lower.69ex\hbox{$\cdot$}\kern.3em}
 {\wedge\kern-.46em\lower.58ex\hbox{$\cdot$}\kern.25em}
 {\wedge\kern-.38em\lower.5ex\hbox{$\cdot$}\kern.1em}
 {\wedge\kern-.3em\lower.5ex\hbox{$\cdot$}\kern.1em}}\limits}

\def\simge{\mathrel{%
   \rlap{\raise 0.511ex \hbox{$>$}}{\lower 0.511ex \hbox{$\sim$}}}}

\def\simle{\mathrel{
   \rlap{\raise 0.511ex \hbox{$<$}}{\lower 0.511ex \hbox{$\sim$}}}}

\reversemarginpar

\newcommand{\cmpmlt}{composition multiset\xspace}

\newcommand{\ab}{\Sigma}

\newcommand{\numwt}{\ell}
\newcommand{\nummax}{L}

\newcommand{\pd}[2][\empty]{\dfrac{\partial #1}{\partial #2}}

\newcommand{\df}[1][\rm{def}]{\hspace{0.3ex}\stackrel{#1}{=}\hspace{0.3ex}}
\newcommand{\comment}[1]{}

\renewcommand{\hat}{\widehat}

\newcommand{\strs}{{s}}
\newcommand{\strt}{{t}}

\newcommand{\cnfset}{E}
\newcommand{\cnfmax}{E}

\newcommand{\rP}{P^*}
\newcommand{\rQ}{Q^*}
\newcommand{\rA}{A^*}
\newcommand{\rB}{B^*}
\newcommand{\rC}{C^*}

\title{String Reconstruction from Substring Compositions}

\author{Jayadev Acharya\footnote{UCSD. email: \texttt{jacharya@ucsd.edu}} \and
Hirakendu Das\footnote{Yahoo Labs!. email: \texttt{hdas@yahoo-inc.com}} \and
Olgica Milenkovic\footnote{UIUC. email: \texttt{milenkov@uiuc.edu}.} \and
Alon Orlitsky\footnote{UCSD. email: \texttt{alon@ucsd.edu}.} \and
Shengjun Pan\footnote{Google. email: \texttt{s1pan@eng.ucsd.edu}.}
}

\begin{document}

\begin{titlepage}
\maketitle\thispagestyle{empty}
 \begin{abstract}
Motivated by mass-spectrometry protein sequencing, we consider
a simply-stated problem of reconstructing a string from the
multiset of its substring compositions.
We show that all strings of length 7, one less than a prime,
or one less than twice a prime, can be reconstructed uniquely up to reversal.
For all other lengths we show that reconstruction is not
always possible and provide sometimes-tight bounds
on the largest number of strings with given substring compositions.
The lower bounds are derived by combinatorial arguments
and the upper bounds by algebraic considerations
that precisely characterize the set of strings with the same substring
compositions in terms of the factorization of bivariate polynomials.
The problem can be viewed as a combinatorial simplification
of the turnpike problem, and its solution may shed light
on this long-standing problem as well.
Using well known results on transience of multi-dimensional random walks, we 
also provide a reconstruction algorithm 
that reconstructs random strings over alphabets of
size $\ge4$ in optimal near-quadratic time. 
\end{abstract}
\end{titlepage}
\section{Motivation}
A protein is a long sequence of amino acids whose composition
and order determine the protein's properties and functionality.
A common tool for finding the amino-acid sequence is 
\emph{mass spectrometry}~\cite{T92,M01}.
It takes a large number of identical proteins, ionizes and randomly breaks
them into substrings, and analyzes the resulting mixture 
to determine the substring weights.
The substring weights 
are then used to infer the amino-acid sequence.

We make two simplifying assumptions that reduce
protein mass-spectrometry reconstruction to a simply-stated combinatorial problem
that we then analyze.

{\it Assumption 1:}
The composition of every protein substring can be deduced from its weight.
For example, let $A$, $B$, and $C$ be three amino acids with
respective weights 13, 7, and 4.
A string weight of 11 clearly consists of one $B$ and one $C$.
Similarly, weight 18 implies two $B$'s and one $C$.
However, weight 20 could arise from
one $A$ and one $B$ or from 5 $C$'s, hence we cannot deduce
the composition from the weights.
The assumption states that such confusions never arise.

{\it Assumption 2:} 
Each protein-sequence bond gets cut independently with the same probability.
For example, if the sequence is $ABC$ and the cut probability
is $p$, then 
the partition $A|B|C$ is obtained with probability $p^2$,
the partitions $A|BC$, and $AB|C$ are obtained with probability $p(1-p)$, 
and the partition $ABC$ with probability $(1-p)^2$. 
 
While both assumptions are clearly idealized, one can imagine scenarios
where they roughly hold. For example, the first assumption holds 
if all amino-acid weights are sufficiently large and different.
More detailed assumptions, such as some amino-acid weight
similarities, or unequal cut probabilities, can be considered
as well, see \eg Section~\ref{sec:concl-extens}, 
but the two current assumptions provide a clean simple-to-analyze formulation.
Additionally, independently of their precise validity, these
assumptions convert protein sequencing to a very simple
string-reconstruction problem that is interesting on its own merit.
It is also a combinatorial simplification of the long-open
\emph{turnpike problem}~\cite{D00, SSL90}, and may provide insight into
the structure of its solutions.



The \emph{composition} of a string is the multiset of its elements,
namely the number of times each element appears, regardless of
the order. Compositions of strings also known as \emph{Parikh vectors}.
To derive the reconstruction problem, observe that 
the first assumption implies that the composition, and hence also the length,
of each substring can be determined.
The second assumption implies that, ignoring
small effects at the sequence ends, all substrings of a given length
would appear roughly the same number of times.
For example, for cut probability $p$, strings of length $k$ would
appear a number proportional to $p^2(1-p)^{k-1}$.
Since the substring lengths can be determined, their number of
appearances can be normalized so that the composition of each substring appears exactly once.
The problem then is to reconstruct a length-$n$ protein sequence from the multiset
of all its $\binom{n+1}2$ substring compositions.

For example, the composition of the string $BABCAA$ 
is the multiset $\sets{A,A,A,B,B,C}$, denoted
$A^3B^2C$ to indicate that the string consists of
three $A$'s, two $B$'s, and  one $C$.
Sequencing say the string $ACAB$ would result in ${4+1\choose 2}=10$
substring compositions: $A$, $A$, $B$, $C$, $AB$, $AC$, $AC$, $A^2C$, $ABC$,
and $A^2BC$.
Note that for each of substring, only the composition, and not the order, is given,
and that the compositions are provided  along
with their multiplicity, but not the location they appear in the string.
To reconstruct the original string $ACAB$ from its substring
compositions,
 note that the compositions imply it consists
of two $A$'s, a $B$, and a $C$, and the two $A$'s don't appear together,
hence the string is $ABAC$, $ACAB$, $ABCA$, or their reversals.
The appearance of the substring composition $A^2C$ implies $ACAB$ or its reversal.

Clearly a string and its reversal have the same composition multiset
and hence cannot be distinguished. We therefore attempt to recover
the string up to reversal.

As a final simplification, note that protein sequences are over
an alphabet consisting of the 20 amino acids.
Yet reconstruction of binary strings extends to reconstructing strings over any finite alphabet.
For example, to reconstruct the string $ACAB$ above, first
replace appearances of $A$'s by 1 and of $B$'s and $C$'s by 0.
This yields the compositions 0, 0, 1, 1, 01, 01, 01, $0^21$, $01^2$,
$0^21^2$,
which imply 1010, namely $A$'s appear in the first and third
locations. Then replace $A$'s and $B$'s by 1 and $C$'s by 0's
to see that $B$ appears in the last location, and finally deduce
$C$'s location. A short argument
shows that this recursive reconstruction using binary strings
always works.


We therefore consider the reconstruction of binary strings from their
composition multisets.

\ignore{
The assumption implies that in the mixture, each substring
will appear the same number of times, hence for example 
$A$ will occur twice as many times as $B$, $AB$, $AA$ and $AAB$.
This assumption too cannot be expected to hold precisely, but
if each cut occurs with the same probability, then as the number
of proteins tested increases, all substrings of a given length
will appear roughly the same number of times, and that number
could be adjusted by the length to correspond to a single appearance.
%
Under the first assumption, mass-spectrometry protein sequencing
reduces to reconstructing a string from a collection of its
substring compositions where each substring composition is
given an unknown number of times.
With the second assumption, the problem is further reduced to 
reconstructing a string from the collection of its substring
composition, each given exactly once.
We study the extent to which that can be done.
}

\section{Definition and results}
The \emph{\cmpmlt} of a string $\strs=s_1s_2\ldots s_n\in\setzo^n$ is
the multiset 
\[
\cS_\strs
\ed
\sets{
\sets{s_i,s_{i+1}\upto s_{j}}: 1\le i\le j\le n
}
\]
of compositions of all ${n+1\choose2}$ contiguous substrings of $\strs$.
For example, 
\[
\cS_{001}  =  \{0, 0, 1,  0^2, 01,  0^21 \}
\quad\text{and}\quad
\cS_{010}  =  \{ 0,  0,  1,  01,  01,  0^21 \}.
\]
Note that the number of substrings with any given composition is
reflected in the multiset, but their locations are not.

Two strings $\strs$ and $\strt$ are \emph{equicomposable},
denoted $\strs\sim\strt$, if they have the same composition multisets.
Equicomposability is clearly an equivalence relation.
The \emph{reversal} of a string  $\strs=s_1s_2\ldots s_n$ is
the string $\strs^*\ed s_ns_{n-1}\ldots s_1$.
A string and its reversal are clearly equicomposable, but this is not a
problem as in our original motivation
they represent the same protein, hence are effectively the same.
We say that a string is \emph{reconstructable} if it is equicomposable
only with itself and its reversal, hence can be determined
up to reversal from its composition multiset.
A string $\strs$ is \emph{non-reconstructable} or \emph{confusable} if it is equicomposable with
another string $\strt\ne\strs^*$, and if so, we also call $\strs$ and
$\strt$ \emph{confusable}.

As we did with $ACAB$, it is easy to verify that all binary strings of length at most 7 are
reconstructable. 
A natural question is therefore whether all binary strings are reconstructable.
In Section~\ref{S_smp_cnf} we show that for length 8,
\[
01001101\sim 01101001.
\]
Since the two strings are not reversals of each other, they are confusable and neither is reconstructable.
An extension of the example yields non-reconstructable binary strings of length $n$
whenever $n+1$ is a product of two integers, each at least 3.
For example, of length 8 as $8+1=3\cdot 3$, length 11 as $12=4\cdot 3$, etc.
That leaves the question as to whether all binary strings of the
remaining lengths are reconstructable.
Observe that these remaining lengths $n$ are precisely those for which
$n+1$ is either 8, a prime, or twice a prime.
One of the results we prove is that indeed all strings of these
lengths are reconstructable.

To do so, in Sections~\ref{S_pln_rpr} and~\ref{section:rcp_ply} we represent a
binary string $\strs$ by a \emph{generating} bivariate polynomial $P_\strs\in\ZZ[x,y]$.
We also define the \emph{reciprocal} $\rP$ of polynomial $P$.
We show that two strings $\strs$ and $\strt$ are equicomposable if and only if
\[
P_\strs\rP_\strs
=
P_\strt\rP_\strt.
\]
We also consider the sets
\[
\cnfset_\strs
\ed
\sets{\strt: \strt\sim\strs}
\]
and
\[
\mathcal{P}_{\cnfset_\strs}
\ed
\sets{P_{\strt}: \strt\sim \strs}
\]
of strings and generating polynomials equicomposable with a string
$\strs$  and show that it has a very simple polynomial
characterization. Let
\[
P_\strs=
P_1 P_2\cdots P_k
\]
be the prime factorization of $P_\strs$ over $\ZZ[x,y]$. 
It is perhaps intriguing that the equicomposable set of $\strs$ 
can be expressed exactly and simply as:
\[
\mathcal{P}_{\cnfset_\strs}
=
\sets{\tilde P_1 \tilde P_2\cdots \tilde P_k:\textit{each $\tilde P_i$ \textit{is} $P_i$ \textit{or} $\rP_i$}}.
\]

In Section~\ref{S_unq_rcn}, we use these results to 
show that indeed if $n+1$ is 8, a prime, or twice a prime, all length-$n$
binary strings are reconstructable.
Along with Section~\ref{S_smp_cnf}'s construction, this establishes
the exact lengths for which all strings are reconstructable.

In Sections~\ref{S_lwr_bnd} and~\ref{S_upr_bnd} we consider lengths where reconstruction is
not always unique.
Let 
\[
\cnfmax_n
\ed
\max\sets{|\cnfset_\strs|: \strs\in\setzon}
\]
be the the largest number of mutually-equicomposable $n$-bit strings.
Since every string is equicomposable with its reversal, $\cnfmax_n\ge 2$ for
every $n\ge 2$, and from the above, $\cnfmax_n\ge 4$ whenever $n+1$ is
a product of two integers $\ge 3$.

Generalizing the interleaving construction, we lower-bound $\cnfmax_n$, and using
the polynomial representation and results on cyclotomic-polynomials, we upper bound it.
Considering the distinct-prime factorization $n+1=2^{e_0}p_1^{e_1}p_2^{e_2}\ldots p_k^{e_k}$,
we show that
\[
2^{\Floor{\frac{e_0}{2}}+e_1+e_2+\cdots+e_k}\le \cnfmax_n\le \min\{2^{(e_0+1)(e_1+1)\cdots(e_k+1)-1}, (n+1)^{1.23}\},
\]
where $(e_0+1)(e_1+1)\cdots(e_k+1)=d(n+1)$ is the number of divisors of
$n+1$.
We also show that when $n+1$ is a prime power or twice a prime power, the
lower bound is tight. For all $k\ge1$,
\[
\cnfmax_{2^k-1}= 2^{\Floor{\frac{k}{2}}},
\]
and for prime $p\ge3$,
\[
\cnfmax_{p^k-1}=
\cnfmax_{2p^k-1}=2^k.
\]
It follows that for $n=3^k-1$, 
\[
\cnfmax_n=(n+1)^{\log_3 2}.
\]

We then provide a backtracking algorithm that reconstructs
a string from it substring compositions.  Using some well
known results on random walks in high dimensions, 
we show that the algorithm with high probability reconstructs
a string in time $O(n^2\log n)$ for alphabet size $\ge4$. 
This is superior, albeit random, compared to the performance
of polynomial factorization algorithms. 
The algorithm can reconstruct strings over alphabet 
size 2 and 3, although with weaker guarantees. 

\section{Relation to other work}
String reconstruction from multiset decompositions is related to two types of problems.
By its formulation it is similar to other string-reconstruction problems,
while mathematically it is closer to the well-known turnpike problem.

Several string-reconstruction variations have been previously considered.
Reconstruction of a string from a few substrings was considered in~\cite{MS95}.
Reconstruction of a string from its subsequences (not necessarily contiguous)
was considered in~\cite{L01,DS03,BKKM04,VS08}. 
In these problems however, the substrings or subsequences
themselves, which include the order of their symbols, are given.
By contrast, in our problem, for each substring we are
given just the composition, neither the order of the symbols
within it nor the substring's location in the original string.

Note that four simple problems of reconstructing a string from its
substrings can be formulated.
They differ by whether
or not we are given
(a) the order of the substrings in the string, and (b) the order of the bits in each substring.
All four problems are of interest when only some of the
$\binom{n+1}2$ substrings are given. But when all substrings are given,
knowing the order of either the substrings or their symbols clearly determines the
original string.
It is only when, as in this mass-spectrometry application,
neither the substring order nor the symbol order are provided, that
the reconstruction question arises.
Some partial results on this problem were presented in~\cite{ADMOP10}.


In terms of the solutions and proof techniques,
our problem is more closely related to the well-known 
\emph{turnpike problem} where the locations of $n$ highway exits
need to be recovered from the multiset of their ${n\choose2}$
inter-exit distances.
For example inter-exit distances 1, 2, 3, 3, 5, 6 correspond to exit locations 0, 1, 3, 6.

The turnpike problem originally arose in X-Ray crystallography~\cite{P35, P44}
and has many applications, including to DNA analysis~\cite{DK88}.
It is also of theoretical interest as it has an algorithm
whose run time is polynomial in the largest inter-exit distance, but
it is not known whether it has an algorithm whose run time is
polynomial in $n$, independently of the distances.
Several variation of the problem have been recently considered,
e.g.~\cite{DBLP:journals/tcs/DauratGN05,CHK09} and references therein.
A related problem of characterizing strings that have the same
\emph{set} of compositions, instead of multisets, is
considered in~\cite{FiciL11}.

String reconstruction from substring multisets is a combinatorial
simplification of the turnpike problem~\cite{UstoSK,CHK09}.
To see that, convert every string-reconstruction instance to a
turnpike problem whose solution implies the original string.
Given an $n$-bit string reconstruction problem, replace each substring
composition by an inter-exit distance obtained by replacing each 0 by 1 and each
1 by $n+1$ and summing the values.
Then find the exit locations and replace back the bit values.

For example, to recover the string 1011 from the compositions
0, 1, 1, 1, 01, 01, 11, $01^2$, $01^2$, $01^3$, replace every 0 by 1 and every
1 by $4+1=5$ to obtain inter-exit distances 1, 5, 5, 5, 6, 6, 10,
11, 11, 16 which in turn correspond to locations 0, 5, 6, 11, 16.
Then convert 1 back to 0, and 5 back to 1 to obtain the original string
1011.
Note that the mapping of 1 to $n+1$ prevents spurious solutions.


This reduction shows that string reconstruction can be solved
in polynomial time~\cite{LLL82}, but this is not our main focus. Instead, we
show that several questions unsolved for the turnpike problem
can be answered for string reconstruction. Hence, in addition to the
problem's intrinsic value, and potential application to protein
reconstruction, its solution may provide
useful insights for the general turnpike problem.

For example, some of the more important unsolved questions
about the turnpike problem concern the number of solutions
a given instance may have~\cite{SSL90}.
A turnpike problem with $n\ge 6$ exits may have 
multiple solutions, but their largest possible number is not known for
any such $n$. By contrast, for
string reconstruction we show that when $n+1$ is 8, a
prime, or twice a prime reconstruction is always unique.
We also determine the exact number of maximal reconstructions
whenever $n+1$ is a prime power or twice a prime power, and
conjecture that the same formula holds for all string lengths.

We also note that some of the polynomial techniques we
use are related to those applied to the turnpike problems~\cite{RS82,  SSL90}.
Yet others, such as the bivariate polynomial formulation
and relation to cylcotomic polynomials, seem new.


\section{Simple confusions}
\label{S_smp_cnf}
A simple computer search shows that all binary strings of length at most 7 are reconstructable.
For length 8, the strings 01001101 and 01101001 are confusable and are
not reversals of each other. To see that, note that they
can be parsed  as \textbf{01}\,0\,\,\textbf{01}\,1\,\textbf{01}
and  \textbf{01}\,1\,\textbf{01}\,0\,\,\textbf{01}. 
Both have a common substring \textbf{01}, which is
interleaved with 0\,1 for the first string and with its reversal 1\,0
in the second string. 
The \emph{interleaving} of string $\strs$ with the bits of $\strt=t_1\ldots t_m$ is the string
$\strs\circ \strt\ed\strs\, t_1\,\strs\, t_2\,\dotsb \,t_m\,\strs$.
The first string above is therefore $01\circ01$ while the second is $01\circ10$.

Recall that $\strt^*$ represents the reversal of $\strt$, and that
$\sim$ indicates that two strings are equicomposable.
We show that $\strs\circ\strt\sim\strs\circ\strt^*$ for any $\strs$ and $\strt$.
We demonstrate that fact when $\strt$ has length 3, and the general result follows similarly.
The interleaving $\strs\circ\strt$ is represented by the figure below,
where each triangle represents the string $\strs$.

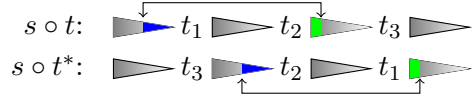
\begin{center}
\def\triangle{[shading=axis,shading angle=90]
      ++(0,0.2)-- ++(1.2,-0.2)-- ++(-1.2,-0.2)--cycle}

\begin{tabular}{r@{:}l}
$\strs\circ \strt$ &\quad
\begin{tikzpicture}[baseline=-2.4pt,scale=0.65]

\foreach \x in {0,1,2,3}
{
  \draw (2*\x,0) \triangle;
  \ifthenelse{\x>0}{\node at (2*\x-0.4,0) {$t_{\x}$};}{}
}
\draw[<->] (0.6,0.2)--(0.6,0.5)--(4.2,0.5)--(4.2,0.2);
\begin{scope}
  \clip\triangle;
  \fill[blue] (0.6,-1) rectangle (2,1);
\end{scope}
\clip (4,0) \triangle;
\fill[green!100] (4.2,-1) rectangle (2,1);
\end{tikzpicture}\\
$\strs\circ \strt^*$ &\quad
\begin{tikzpicture}[baseline=-2.4pt,scale=0.65]
\foreach \x in {0,1,2,3}
{
  \draw (2*\x,0) \triangle;
  \ifthenelse{\x>0}{\node at (7.6-2*\x,0) {$t_{\x}$};}{}
}
\draw[<->] (2.6,-0.2)--(2.6,-0.5)--(6.2,-0.5)--(6.2,-0.2);
\begin{scope}
  \clip (2,0) \triangle;
  \fill[blue] (2.6,-1) rectangle (3.2,1);
\end{scope}
\clip (6,0) \triangle;
\fill[green] (6,-1) rectangle (6.2,1);
\end{tikzpicture}
\end{tabular}


\end{center}

Clearly any substring of $\strs\circ\strt$ can be uniquely
mapped to a substring of $\strs\circ\strt^*$ with the same composition.
For example, as illustrated in the figure, the substring in
$\strs\circ\strt$
consisting of a tail of $\strs$ (blue), $t_1$, $\strs$, $t_2$,
and a head of $\strs$ (green) has the same composition as the substring
in $\strs\circ\strt^*$ consisting of the same tail of $\strs$ (blue),
$t_2$, $\strs$, $t_1$, and the same head of $\strs$ (green).
Thus $\strs\circ\strt$ and $\strs\circ\strt^*$ have the same
multiset of compositions, and hence are equicomposable.

If the length of $\strs$ and $\strt$ are $m$ and $m'$
respectively, then the length of $\strs\circ \strt$
is $n=(m+1)(m'+1)-1$. If $m\ge2$, $m'\ge2$,
we can always choose non-palindromic $\strt$, namely $\strt\ne \strt^*$, ensuring that
$\strs\circ \strt\ne \strs\circ \strt^*$,
and non-palindromic $\strs$, namely $\strs\ne\strs^*$, ensuring that 
$\strs\circ \strt\ne(\strs\circ \strt^*)^*$.
It follows that whenever $n+1$ is a product of
two integers, each $\ge3$, there are confusable $n$-bit strings.

\section{Polynomial representation}
\label{S_pln_rpr}
The previous section described some confusable strings.
To further characterize confusability, we represent strings as polynomials.
A similar representation has been used for the turnpike
problem~\cite{RS82, SSL90},
however the polynomials used there are univariate, whereas 
bivariate polynomials are better suited for string reconstruction.

We use this representation to show that 
the strings equicomposable with a string $\strs$ can be determined by
the prime factorization of the polynomial representing $\strs$, and
that $\strs$ can be
reconstructed from its composition multiset by polynomial
factorization.

All polynomials and factorizations in the paper are over $\ZZ[x,y]$.
A polynomial whose nonzero coefficients are all 1 is \emph{0-1}.
A term $x^ay^b$ is a \emph{monomial} with $x$-\emph{degree} $a$,
$y$-\emph{degree} $b$, and \emph{total degree} $a+b$.
The $x$-, $y$-, and \emph{total-degrees} of a polynomial are
the corresponding highest degrees of any of its monomials.
For example, for  $x^2y+xy^3$ they are 2, 3, and 4 respectively.

When representing strings by polynomials, $0$ is denoted by $x$ and
$1$ by $y$.
Also, $a_i$ denotes the number of 0's in the first
$i$ bits of a string and $a$ denotes their total number, while
$b_i$ and $b$ represent the same for the 1's. For example,
01011 has $a_3=a=2$, $b_3=1$, and $b=3$.

The \emph{generating polynomial} of a binary string
$\strs = s_1 s_2\cdots s_n$ is
\[
P_\strs(x,y)
\ed
\sum\limits_{i=0}^{n}x^{a_i}y^{b_i}
=
\sum\limits_{i=0}^{n}x^{a_i}y^{i-a_i}.
\]
For example,
\[
P_{0100}(x,y) = 1 + x + xy + x^2y + x^3y.
\]

Generating polynomials of $n$-bit strings are characterized by the following
sufficient properties.
\begin{description}
\item[G1]
They are 0-1.
\item[G2]
They have $n+1$ terms, exactly one of each total degree $0, 1, \ldots, n$.
\item[G3]
For all $1\le i\le n$, the ratio between the terms of total degrees
$i$ and $i - 1$ is either $x$ or $y$.
\end{description}
For example, $1+x+xy+x^2y+x^3y$ is generating, 
but $1-x$, $1+x^2$, and $1+x+y^2$ are not.

Similar to strings, a composition $0^{a}1^{b}$ is represented by the monomial
$x^{a}y^{b}$, and $\cS_\strs(x,y)$ is the summation of all monomials
corresponding to the substring compositions of $\strs$.
For example,
\[
\cS_{0100}(x,y)
= 
\sets{0,0,0,1,0^2,01,01,0^21,0^21,0^31}
=
3x + y + x^2 + 2xy + 2x^2y + x^3y.
\]

Note that for all $i\le j$, 
\[
x^{a_j}y^{b_j}\cdot(x^{a_i}y^{b_i})^{-1}
=
x^{a_j-a_i}y^{b_j-b_i}
\]
represents the composition of the substring $s_{i+1}\ldots s_j$, hence the
generating polynomial provides a simple expression for the composition multiset,
\[
P_\strs(x,y) P_\strs\Paren{\frac1x,\frac1y} = n+1 +
\cS_{\strs}(x,y) + \cS_{\strs}\Paren{\frac1x,\frac1y}.
\]
For instance, in our running example,
\begin{align*}
P_{0100}(x,y)P_{0100}\left(\frac1x, \frac1y\right)
&=
5 + 3x + y + x^2 + 2xy + 2x^2y + x^3y + \frac{3}{x} + \frac{1}{y} + \frac{1}{{x}^{2}} + \frac{2}{xy} +
 \frac{2}{{x}^{2}y} + \frac{1}{{x}^{3}y}\\
&=
5+\cS_{0100}(x,y)+\cS_{0100}\Paren{\frac1x,\frac1y}.
\end{align*}

To reconstruct a string from a composition multiset $\cS$,
we therefore find all strings $\strs$ whose generating polynomial $P_\strs$
satisfies the above equation.
We also have the following basic lemma.
\begin{lem}
Strings $\strs$ and $\strt$ are equicomposable iff
\label{lemma:pxpoox}
\begin{equation}
\label{equation:pxpoox}
P_\strs(x,y) P_\strs\Paren{\frac1x, \frac1y}
=
P_\strt(x,y) P_\strt\Paren{\frac1x, \frac1y}.
\end{equation}
\end{lem}

We next show that for a string $s$, an equation
like~\eqref{equation:pxpoox}
holds only for generating polynomials.
The approach is similar to~\cite{RS82, SSL90} with additional
work needed to prove that the polynomial
obtained satisfied all the properties of 
generating polynomials.

\begin{lem}
\label{lemma:q_generating_oox}
If $P$ is a generating polynomial and $Q\in\ZZ[x,y]$ with
$Q(0,0)>0$ and
\begin{equation}
\label{equation:q_generating_oox}
P(x,y)P\Paren{\frac1x,\frac1y}=Q(x,y)Q\Paren{\frac1x,\frac1y},
\end{equation}
then $Q(x,y)$is a generating polynomial.
\end{lem}
\begin{proof}
We first show that $Q$ satisfies \textbf{G1}.

\noindent\textbf{G1}
Evaluating~\eqref{equation:q_generating_oox} at $x=y=1$, we obtain
$P(1,1)^2 = Q(1,1)^2$, namely $P(1,1) = \pm Q(1,1)$. Suppose `$+$'.
Let $P(x,y) = \sum_{i,j} p_{i,j}x^iy^j$ and $Q(x,y) = \sum_{i,j}
q_{i,j}x^iy^j$, where $p_{i,j},q_{i,j}\ne0$, then,
\[
 \sum_{i,j} p_{i,j} = \sum_{i,j} q_{i,j}.
\]
Comparing the constant term in~\eqref{equation:q_generating_oox},
\[
 \sum_{i,j} p_{i,j}^2 = \sum_{i,j} q_{i,j}^2.
\]
Subtracting the two equations,
\[
 \sum_{i,j} p_{i,j}(1-p_{i,j}) = \sum_{i,j} q_{i,j}(1-q_{i,j}).
\]
Since all $p_{i,j}$ are 1, the left and hence the right sides are $0$.
For all integers $i$, $i(1-i)\le 0$ with equality iff $i$ is either 0 or 1.
Since all $q_{i,j}$ are nonzero integers, they must be 1. The case
of `$-$' is similar. 

\noindent\textbf{G2}
Each term in $P(x,y)P(1/x,1/y)$ is of the form
$x^{a_j-a_i}y^{b_j-b_i}$. Thus the exponents of $x$ and $y$ 
cannot have different signs. $Q$ cannot have two terms of the
same total degree as their
corresponding product in $Q(x,y)Q(1/x,1/y)$ would yield a term 
with $x$- and $y$-degrees of opposing signs, which by G1 would not
get cancelled by any other product term. Let $0=d_0< d_1< \ldots< d_n$
be the degrees of the terms in $Q$. 
The largest degree term in~\eqref{equation:q_generating_oox}
on the LHS is $x^{a_n}y^{n-a_n}$ whose
degree is $n$. The largest
degree on the RHS is $d_n-d_0=d_n$.
So $d_i=d_0+i=i$ and property G2 is proved.

\noindent\textbf{G3}
Consider the linear terms (constant times $x$ or
$y$) in~\eqref{equation:q_generating_oox}.
For any 0-1 polynomial $F(x,y)$, the sum of the coefficients of the linear terms 
in $F(x,y)F(1/x,1/y)$ is the number of pairs of terms in $F$ whose
ratio is $x$ or $y$.
$P$ has $n$ such pairs, hence so does $Q$.
By G2, $Q$ has $n+1$ terms one of each total degree $0, 1,\dotsc,n$, the
ratio
between any two consecutive terms must be $x$ or $y$.
\end{proof}

The two lemmas characterize the generating
polynomials of all the strings in $E_s$.
\begin{thm}
\label{thm:generating_Es}
\[
\mathcal{P}_{\cnfset_\strs}=\{P(x,y)\in\ZZ[x,y]:P(0,0)>0, P(x,y)
P\Paren{\frac{1}{x},\frac{1}{y}}=P_s(x,y)P_s
P\Paren{\frac{1}{x},\frac{1}{y}}\}
\]
\end{thm}

\section{Reciprocal polynomials}
\label{section:rcp_ply}
To apply existing results on polynomial factorization, we
relate $P(1/x,1/y)$ to standard polynomials.
Let $\deg_xP$ and $\deg_yP$ be the highest $x$- and $y$-degrees
of a bivariate polynomial $P$.
The \emph{reciprocal} of $P$ is the polynomial
\[
\rP(x,y) \df x^{\deg_x P}y^{\deg_y P} P\Paren{\frac{1}{x},\frac{1}{y}}.
\]
For example,
\[
(y + 3xy - 2y^2)^*
= 
xy^2\Paren{\frac1y+\frac3{xy}-\frac2{y^2}}
= 
xy + 3y - 2x.
\]
We will use the following reciprocal-polynomial properties.
\begin{description}
\item[R1]
$\deg_x \rP\le\deg_x P$ with equality iff $x\nmid P$, and similarly for $y$.
\item[R2]
The reciprocal of the product is the product of the reciprocals: $(P_1P_2)^*
=
\rP_1\rP_2$.
 \item[R3]
The reciprocal of the generating polynomial of $\strs$ generates the
reversal of $\strs$:
$
\rP_{\strs}
=
P_{\strs^*}.
$
\end{description}
For example,
\[
\rP_{0100}(x,y)
=
(1+x+xy+x^2y+x^3y)^*
=
1+x+x^2+x^2y+x^3y
=
P_{0010}(x,y).
\]

These properties imply a polynomial formulation of
Lemmas~\ref{lemma:pxpoox}
and~\ref{lemma:q_generating_oox} and Theorem~\ref{thm:generating_Es}.
The proofs are almost identical to the proofs of their counterparts
shown earlier and are hence omitted.
\begin{lem}
\label{lemma:pxprx}
Strings $\strs$ and $\strt$ are equicomposable iff
\[
P_\strs \rP_\strs = P_\strt \rP_\strt.
\]
\end{lem}


The second lemma has actually a slightly cleaner formulation than its $1/x$
counterpart.

\begin{lem}
\label{lemma:q_is_gen2}
If $P$ is a generating polynomial and for $Q\in\ZZ[x,y]$
\[
P\rP = Q \rQ,
\]
then $Q$ or its negation is also a generating polynomial.
\end{lem}

\begin{thm}
\label{thm:generating_Es_R}
\[
\mathcal{P}_{\cnfset_\strs}=\{P(x,y)\in\ZZ[x,y]:P(x,y)>0, P(x,y)
P^*(x,y)=P_s(x,y)P_s^*(x,y)\}
\]
\end{thm}

So far, we studied products of generating polynomials and their reciprocals.
The next lemma, addresses the factors of generating polynomials.

\begin{lem}
\label{lemma:q_is_arb}
 $\strs\sim\strt$ if and only if for some $A,B\in\ZZ[x,y]$,
\[
P_\strs = A B
\qquad\text{and}\qquad
P_\strt = A \rB .
\]
\end{lem}
\begin{proof}
Let $A = \gcd{P_\strs}{P_\strt}$, then $P_\strs = A B$ and $P_\strt = A C$ for relatively
prime $B$ and $C$. We show that $C=\rB$.

By Lemma~\ref{lemma:pxprx}, $P_\strs \rP_\strs = P_\strt \rP_\strt$.
By R2, $\rP_\strs = \rA \rB$ and $\rP_\strt=\rA\rC$,
hence,
$
B \rB = C \rC.
$
Since $B$ and $C$ are relatively prime, $C$ divides $\rB$.
But since $P_\strs$ and $P_\strt$ have the same total degree,
$B, \rB, C$, and $\rC$, all have the same total degree.
And since $P_\strs$ and $P_\strt$ are generating polynomials, the four
polynomials have constant term 1.
Hence $C=\rB$.

Conversely if $P_s=AB$ and $P_t=A\rB$ then it is easy to see
that $P_sP_s^*=(AB)(AB)^*=AA^*BB^*=AB^*(AB^*)^*=P_tP_t^*$.
\end{proof}

For example, we saw that the 8-bit strings $\strs=01001101$
and $\strt=01101001$ are confusable.
Indeed
$P_{\strs}(x,y)=(1+x+xy)(1+x^2y+x^3y^3)$
and 
$P_{\strt}(x,y)=(1+x+xy)(1+xy^2+x^3y^3)=
(1+x+xy)(1+x^2y+x^3y^3)^*$.
Here, $A$ is a generating polynomial, but this is not always the case.
The 23-bit strings 01000101010000100011001 and 01010100010000110010001
are also confusable,
but here, $A=1+y+xy+xy^2+x^4y^8+x^4y^9+x^4y^{10}+x^5y^{10}$
and $B=1+xy^3+x^3y^5$.
In the above examples, $A$ and $B$ are 0-1 polynomials, but that too
is not always the case. The string 01001001001 can be decomposed
as $(1+x+x^3y-x^4y^2+x^3y^3+x^5y^3+x^5y^4)(1+xy)$
where the first factor has a negative coefficient.
We will return to the structure of these factors in Section~\ref{S_lwr_bnd}.

The previous results imply a very simple expression for
$\cnfset_\strs$,
the set of strings equicomposable with $\strs$.
Consider the factorization of $P_s$ into irreducible factors (prime factorization)
\[
P_\strs=
P_1 P_2\cdots P_k
\]
The next theorem says that every product of the $P_i$'s or their reciprocals
is a generating polynomial and these generating polynomials
exactly correspond to all strings equicomposable with~$\strs$.
\begin{thm}
\label{thm:plycmp}
For every $\strs$,
\[
\mathcal{P}_{\cnfset_\strs}
=
\sets{\tilde P_1 \tilde P_2\cdots \tilde P_k:\textit{each $\tilde P_i$ \textit{is} $P_i$ \textit{or} $\rP_i$}}.
\]
\end{thm}
\begin{proof}
If $P=\tilde P_1 \tilde P_2\cdots \tilde P_k$, where each
$\tilde P_i$ is either $P_i$ or $\rP_i$, let $A$ be the product of the
unmodified $P_i$'s and $B$ be the product of the $P_i$'s that are reciprocated.
Then $P_\strs=AB$ and $P=A\rB$.
Hence $P_\strs\rP_\strs=P\rP$, and by Lemma~\ref{lemma:q_is_gen2},
$P=P_\strt$ for some string $\strt$.
But then $P_\strs\rP_\strs=P_\strt\rP_\strt$, and by Lemma~\ref{lemma:pxprx}, $\strt\in\cnfset_\strs$.
 
Conversely, if $t\in\cnfset_\strs$, then by Lemma~\ref{lemma:q_is_arb},
$P_\strs=AB$ while $P_\strt=A\rB$.
The prime factorizations $A=P_1P_2\cdots P_j$ and $B=P_{j+1}\cdots P_k$ yield
the prime factorizations $P_\strs=P_1P_2\cdots P_k$ and
$P_\strt=P_1P_2\cdots P_j\rP_{j+1}\cdots \rP_k$,
hence $P_\strt$ can be written as $\tilde P_1\tilde P_2\cdots\tilde P_k$.
\end{proof}

The theorem provides a simple formula for the number of
strings equicomposable with any given string $\strs$.
A polynomial is \emph{palindromic} if it equals its reciprocal.
For example, $1+x+xy+xy^2+x^2y^2$, but not $1+x+xy$ $\left((1+x+xy)^*=1+y+xy\right)$.
Let $\nu_\strs$ be the number of non-palindromic terms in the prime
factorization of $P_\strs(x,y)$.
\begin{cor}
\label{cor:plycmp}
For every string $\strs$,
\[
|\cnfset_s|
\le
2^{\nu_\strs}.
\]
In particular, if $\nu_s=1$, then $\strs$ is reconstructable.\qedhere
\end{cor}
\begin{proof}
The first part follows trivially from the theorem.
The second part follows as $\nu_s=0$ implies $\cnfset_s=\sets{\strs}$,
and $\nu_s=1$ implies $\cnfset_s=\sets{\strs,\strs^*}$.
\end{proof}
In Section~\ref{sz_cnf_set} we will return to this result and show that the
 inequality holds with equality.
\section{Cyclotomic polynomials}
\label{section:ccl_pln}

Theorem~\ref{thm:plycmp} characterizes equicomposability in terms of
prime factorization of generating polynomials.
Factoring bivariate polynomials is hard in general, but
the current analysis is simplified by the fact that any factorization of a generating polynomial $P(x,y)$
implies a factorization of the  univariate polynomial $P(x,x)$ obtained
by evaluating $P(x,y)$ at $y=x$.
From properties G1 and G2 of generating polynomials, if $P(x,y)$
generates an $n$-bit string, then
\[
P(x,x)
=
1+x+x^2+\cdots+x^n
=
\frac{x^{n+1}-1}{x-1}.
\]
Factorizations of 
$x^n-1$, and hence of $x^{n+1}-1$,
have been studied extensively and for completeness we present a 
small part of the literature that we use to analyze confusability.


The $n$ complex numbers $e^{\frac{2\pi i}n},e^{2\frac{2\pi i}n}\upto
e^{n\frac{2\pi i}n}=1$, whose $n$th powers are 1, are the $n$th
\emph{roots of unity}.
They are exactly the roots of $x^n-1$, and therefore 
\[
x^n-1
=
\prod_{\omega} (x-\omega),
\]
where the product ranges over all $n$th roots of unity.
The prime factorization of $x^n-1$ over $\ZZ[x]$ therefore
partitions these $n$ terms
into groups, each multiplying to an irreducible polynomial over $\ZZ[x]$.

For a positive integer $d$, a $d$th root of unity is \emph{primitive} if none of its smaller
positive powers is 1.
For example:
\[
\begin{array}{|c|c|c|c|c|c|c|}
\hline
d& 1 & 2 & 3 & 4 & 5 & 6\\
\hline
\text{primitive $d$th roots of unity} & 1 & -1 & e^{\pm2\pi i/3}& \pm
i & e^{2\pi ij/5}, j=1\upto 4 & e^{\pm2\pi i/6}\\
\hline
\end{array}.
\]
The $d$th \emph{cyclotomic polynomial} is 
\[
\Phi_d(x)
\ed
\prod_\omega (x-\omega)
\]
where the product ranges over all primitive $d$th roots of unity.
For example,
\begin{align*}
\Phi_1(x)&=x-1\\
\Phi_2(x)&=x+1\\
\Phi_3(x)&=(x-e^{2\pi i/3}) (x-e^{-2\pi i/3})=x^2+x+1\\
\Phi_4(x)&=(x-i)(x+i)=x^2+1\\
\Phi_5(x)&=\prod_{j=1}^4(x-e^{\frac{2\pi i j}5})=x^4+x^3+x^2+x+1\\
\Phi_6(x)&=(x-e^{2\pi i/6}) (x-e^{-2\pi i/6})=x^2-x+1.
\end{align*}
While we will not use this fact, note that the degree of $\Phi_d(x)$ is 
the number of integers between 1 and $d$ that are relatively prime with $d$.

Every $n$th root of unity is a primitive $d$th root of unity for some
$1\le d\le n$, and it is easy to see that $d\mid n$.
Hence, for every $n$,
\begin{equation}
\label{eqn:xtnmo_ccl_fct}
x^n-1=\prod\limits_{d|n}\Phi_d(x).
\end{equation}
For example,
\begin{align*}
x^2-1&=(x-1)(x+1)\\
x^3-1&=(x-1)(x^2+x+1)\\
x^4-1&=(x-1)(x+1)(x^2+1)\\
x^5-1&=(x-1)(x^4+x^3+x^2+1)\\
x^6-1&=(x-1)(x+1)(x^2+x+1)(x^2-x+1).
\end{align*}

It is easy to see that Cyclotomic polynomials have integer
coefficients.
Comparing coefficients shows that if
$a(x),b(x)\in\QQ[x]$ and $a(x)b(x)\in\ZZ[x]$,
then $a(x), b(x)\in\ZZ[x]$.
By induction on $n$, assume that for all $d< n$ with $d\mid n$,
$\Phi_d(x)\in\ZZ[x]$.
Then~\eqref{eqn:xtnmo_ccl_fct} implies that
\[
\Phi_n(x)=\frac{x^n-1}{\prod\limits_{d<n, d|n}\Phi_d(x),}
\]
and since different cyclotomic polynomials are relatively prime
(no two cyclotomic polynomials share a root),
$\Phi_n(x)\in\QQ[x]$ and by the fact mentioned above, is $\in\ZZ[x]$.
Gauss showed, \eg~\cite{L02}, that the cyclotomic polynomials are irreducible.
The irreducibility proof of general cyclotomic polynomials is somewhat involved.
However, we will mostly need to factor $x^p-1$ and $x^{2p}-1$ for
a prime $p$. Their factorization require the irreducibility of
$\Phi_p(x)$ and $\Phi_{2p}(x)$, which follow easily from~\emph{Eisenstein's Criterion}.

\begin{lem}[Eisenstein's Criterion]
Let $P(x)=a_nx^n+a_{n-1}x^{n-1}+\dotsb+a_0$.
If some prime $p\mid a_0,a_1\upto a_{n-1}$ but $p\nmid a_n$
and $p^2\nmid a_0$, then $P(x)$ is irreducible over $\ZZ$.
\end{lem}
\begin{proof}
Suppose towards a contradiction that 
$P(x)=(b_kx^k+\dotsb+b_0)(c_\ell x^\ell+\dotsb+c_0)$
where $k,\ell\ge1$ and $b_k,c_\ell\ne0$.
Since $p$ divides $a_0=b_0c_0$ but $p^2$ does not,
$p$ divides exactly one of $b_0$ and $c_0$, without loss of generality, $b_0$.
Since $p$ does not divide $a_n=b_kc_\ell$, there is a smallest index
$i<k<n$ such that $p$ does not divide $b_i$.
It follows that $p\mid a_i,b_0,b_1\upto b_{i-1}$ and $p\nmid c_0,b_i$,
contradicting $a_i=b_ic_0+b_{i-1}c_1+\dotsb+b_0c_i$.
\end{proof}
Eisenstein Criterion's best-known application is to the
irreducibility of the cyclotomic polynomial $\Phi_p(x)$ for a prime $p$.
\begin{cor}
\label{cor:pmo_ird}
For every prime $p$, 
\[
\Phi_p(x)
=1+x+x^2+\dotsb+x^{p-1}
\]
is irreducible over $\ZZ$. 
\end{cor}
\begin{proof}
Substituting $x$ by $x+1$,
\[
\Phi_p(x+1)
=\frac{(x+1)^p-1}{x}
=
\binom p1+\binom p2x+\dotsb+\binom pp x^{p-1}.
\]
Since $p$ is prime, it divides $\binom p1, \binom p2\upto\binom p{p-1}$, but not $\binom pp$.
Therefore $p$ divides all coefficients except the leading.
Furthermore, $p^2$ does not divide the constant coefficient $\binom p1$.
By Eisenstein's Criterion, $\Phi_p(x+1)$ and therefore $\Phi_p(x)$ are irreducible over $\ZZ$.
\end{proof}
Note that $\Phi_{2p}(x)=\Phi_p(-x)$ for odd primes $p$, 
hence the irreducibility of $\Phi_{p}(x)$ suffices for our purposes.


\begin{lem}
\label{lemma:ccl_pln}
$\Phi_d(x)$ is palindromic for all $d\ge 2$.
\end{lem}
\begin{proof}
For any polynomial $P(x)$, the roots of 
$P^*(x)$ over $\CC$ are exactly the reciprocals of the roots of $P(x)$, and
the constant coefficient of $P^*(x)$ is the leading coefficient of
$P(x)$.
Since the reciprocal polynomial always exists, and the roots and constant coefficient specify a polynomial,
these two conditions are also sufficient for a polynomial to be
the reciprocal of $P(x)$.

For every $d$, the roots of $\Phi_d(x)$ are on the unit circle, hence
the reciprocal of a root is also its complex conjugate.
Furthermore, since $\Phi_d(x)$ has real coefficients, every root appears
with its complex conjugate, hence the reciprocals
of the roots are roots too.
Finally, for $d\ge 2$, $\Phi_d(x)$ has the same leading and constant
coefficient, namely 1,
hence it is its own reciprocal.
\end{proof}

\section{The size of confusable sets}
\label{sz_cnf_set}
The next lemma shows that in any factorization of a generating polynomial,
no two terms are reciprocal of each other. Our proof is different
and simpler than the one used in~\cite{SSL90} for the general turnpike problem.
\begin{lem}
Generating polynomials do not have two mutually reciprocal factors.
\end{lem}
\begin{proof}
Let $P(x,y)$ generate an $n$-bit string.
If $A(x,y)A^*(x,y)$ divides $P(x,y)$, then
$A(x,x)A^*(x,x)$ $=A^2(x,x)$ divides $P(x,x)=1+x+x^2+\dotsb+x^n$.
Clearly, $A^2(x,x)$ has double roots over $\CC$. The roots
of $1+x+\cdots+x^n$ are all distinct,
hence it has no multiple roots.
\end{proof}
If follows that Corollary~\ref{cor:plycmp} holds with equality.
\begin{cor}
\label{cor:cnf_non_pln_eql}
For every string $\strs$,
\[
|\cnfset_s|
=
2^{\nu_\strs},
\]
and $\strs$ is reconstructable iff $\nu_s\le 1$.$\hfill$$\qedsymbol$
\end{cor}
For example, for $\strs=01001101$
the factorization of its generating polynomial into irreducible factors
is $P_{\strs}(x,y)=(1+x+xy)(1+x^2y+x^3y^3)$, 
 hence $\nu_s=2$. The four confusable strings
are $E_{01001101}=\{01001101, 01101001, 10110010, 10010110\}$.
\section{Unique reconstruction for prime-related lengths}
\label{S_unq_rcn}
In Section~\ref{S_smp_cnf}, we saw that whenever $n+1$ is a product of
two integers $\ge3$, there are confusable $n$-bit strings.
The remaining lengths are those where $n+1$ is a prime, twice a prime, or 8.
We use the polynomial representation to show that for all these
lengths reconstruction is always unique.

While factorization and irreducibility of $P(x,y)$ is related to that
of $P(x,x)$, the two are not in complete correspondence.
Irreducibility of $P(x,y)$ does not imply irreducibility of $P(x,x)$.
For example, $P(x,y)=2x^2-y^2=(\sqrt2 x+y)(\sqrt 2x-y)$
hence is irreducible over $\ZZ$, yet $P(x,x)=x^2$ is reducible.


Similarly, in the direction we need,
irreducibility of $P(x,x)$ does not generally imply irreducibility of
$P(x,y)$. For example, $P(x,y)=(x-y+1)(x+1)$ is reducible even
though $P(x,x)=x+1$ is irreducible.
Yet the next two lemmas show that for generating polynomials,
and more generally when $\deg P(x,y)=\deg P(x,x)$,
 this implication holds.


\begin{lem}
\label{lemma:sng_max_fct}
Any factor of a generating polynomial has a single term of highest total degree.
\end{lem}
\begin{proof}
Let $\hat P(x,y)$ be the sum of the highest-total-degree terms in a polynomial $P(x,y)$.
For example, if $P(x,y)=x^3+x^2y^2+y^4$ then $\hat P(x,y)=x^2y^2+y^4$.
If $A(x,y)$ is a factor of $P(x,y)$, then $\hat A(x,y)$ divides $\hat P(x,y)$,
and if $P(x,y)$ is also generating, then $\hat P(x,y)$ consists of a
single term, hence so does its factor $\hat A(x,y)$.
\end{proof}

\begin{lem}
\label{lemma:fct_xy_fct_xx}
The number of terms in the prime factorization of a generating
polynomial $P(x,y)$ is at most the corresponding number for $P(x,x)$.
Hence if $P(x,x)$ is irreducible, so is $P(x,y)$.
\end{lem}
\begin{proof}
Let $A_1(x,y)A_2(x,y)\dotsb A_k(x,y)$ be the prime factorization of
$P(x,y)$.
Then $A_1(x,x)\allowbreak A_2(x,x)\allowbreak\dotsb\- A_k(x,x)$ is a factorization of $P(x,x)$
into polynomials that by Lemma~\ref{lemma:sng_max_fct} are non-constant.
\end{proof}

\begin{thm}
\label{thm:lng7}
All strings of length 7 are reconstructable.
\end{thm}
\begin{proof}
A simple computer search proves the result, but we also provide a 
theoretical proof in Appendix~\ref{sec:appendixB}.
\end{proof}
\begin{thm}
\label{thm:prm}
All strings of length one less than a prime are reconstructable.
\end{thm}
\begin{proof}
If $\strs$ is $(p-1)$-bits long then by the properties of generating polynomials,
\[
P_\strs(x,x)=1+x+x^2+\ldots +x^{p-1}.
\]
By Corollary~\ref{cor:pmo_ird}, $P_\strs(x,x)$ is irreducible.
By Lemma~\ref{lemma:fct_xy_fct_xx}, $P_\strs(x,y)$ is irreducible.
By Corollary~\ref{cor:plycmp}, $s$ is reconstructable.
\end{proof}
For example, for $n=2$, both $P_{00}(x,y)=1+x+x^2$ and $P_{01}(x,y)=1+x+xy$, 
and their symmetric versions, $P_{10}(x,y)$ and $P_{11}(x,y)$, are irreducible.

To prove the result for twice a prime, we'll need the following simple lemma.
\begin{lem}
\label{lemma:cnf_non_pln_fct}
If $\strs$ is confusable, then $P_\strs=AB$ for some
non-palindromic polynomials $A$ and $B$.
\end{lem}
\begin{proof}
Let $\strs$ be confusable with $\strt$.
By Lemma~\ref{lemma:q_is_arb}, 
$P_\strs=AB$ and $P_\strt=A\rB$ for some $A$ and $B$.
If $B$ is palindromic, then $P_\strs=AB=A\rB=P_\strt$,
namely $\strs=\strt$, 
while if $A$ is palindromic, then $\rP_\strs=\rA\rB=A\rB=P_\strt$,
namely $\strs=\strt^*$.
\end{proof}
\begin{thm}
\label{thm2p}
All strings of length one less than twice a prime are reconstructable.
\end{thm}
\begin{proof}
Let $\strs$ have length one less than twice a prime.
We show that in any factorization 
\[
P_\strs(x,y)=f(x,y)g(x,y),
\]
at least one of $f(x,y)$ and $g(x,y)$ is palindromic.
By Lemma~\ref{lemma:cnf_non_pln_fct}, $\strs$ is reconstructable. 

Eisenstein's criterion again implies the prime factorization
\begin{align*}
P_\strs(x,x)
&=
1+x+x^2+\dotsb+x^{2p-1}\\
&=
(1+x)(1-x+x^2-\dotsb +x^{p-1})(1+x+\dotsb +x^{p-1}).
\end{align*}
Hence there are only three factorizations of $P_\strs(x,y)$ into two factors. \\[3mm]
\textit{Case 1:}
$f(x,x)=1+x$ and $g(x,x)=(1-x+x^2-\dotsb +x^{p-1})(1+x+\dotsb +x^{p-1})$,\\[3mm]
\textit{Case 2:}
$f(x,x)=1+x+\dotsb +x^{p-1}$ and $g(x,x)=(1+x)(1-x+x^2-\dotsb +x^{p-1})=1+x^p$,\\[3mm]
\textit{Case 3:}
$f(x,x)=1-x+\dotsb +x^{p-1}$ and $g(x,x)=(1+x)(1+x+\dotsb+x^{p-1})$. \\[3mm]
We show that in all three cases, at least one of $f(x,y)$ and $g(x,y)$ is palindromic.\\[3mm]
\textit{Case 1:} 
By Lemma~\ref{lemma:sng_max_fct},
$f(x,y)$ must be either $1+x$ or $1+y$,
and both are palindromic.\\
\textit{Case 2:}
Let $\strs=\strs_1\strs_2\ldots\strs_n$.
Then
$P(x,y)\ed P_\strs(x,y)=\sum_{i=0}^{2p-1}x^{a_i}y^{i-a_i}$.
Denote 
$\pd[P]{x}(x,y)|_{y=x}$ by $ P_x'$
to obtain 
\[
P_x'=\sum_{i=1}^{2p-1}a_ix^{i-1}.
\]
Therefore,
\[
\pd[P]{x}(x,y)=
f(x,y)\pd[g]{x}(x,y)
+
g(x,y)\pd[f]{x}(x,y)
\]
implies
\[
P_x'=(1+x+\dotsb +x^{p-1})g_x'
+(1+x^p)f_x',
\]
and multiplying both sides by $1-x$,
\begin{align*}
(1-x)P_x'
&=(1-x^p)g_x'+[x^p(1-x)+(1-x)]f_x'.
\end{align*}
The left-hand side can be written as
\[
\sum_{i=1}^{2p-1}a_ix^{i-1}
-
\sum_{i=1}^{2p-1}a_ix^{i}
=
a_1+\sum_{i=1}^{2p-2}(a_{i+1}-a_i)x^{i}-a_{2p-1}x^{2p-1}
=
\sum_{i=0}^{2p-2}s_{i+1}x^{i}-a_{2p-1}x^{2p-1}
\]
and the right-hand side as
\[
\underbrace{\left[(1-x)f_x'+g_x'\right]}
_{(I)}
+\underbrace{x^p\cdot
\left[(1-x)f_x'-g_x'\right]}
_{(II)}.
\]
Note that
$\deg(I)\le \max\{\deg f,\deg g-1\}=p-1$,
while all terms in $(II)$ have degree $\ge p$.
Hence,
\begin{align*}
(1-x)f_x'+g_x'
&
=\sum_{i=0}^{p-1}s_{i+1}x^i,\\
(1-x)f_x'-g_x'
&=\sum_{i=p}^{2p-2}s_{i+1}x^{i-p}-a_{2p-1}x^{p-1}
=\sum_{i=0}^{p-2}s_{i+1+p}x^{i}-a_{2p-1}x^{p-1}.
\end{align*}
Subtracting the two equations,
\[
2g_x'
=
\sum_{i=0}^{p-2}(\strs_{i+1}-\strs_{i+1+p})x^{i}+(\strs_{p}+a_{2p-1})x^{p-1}.
\]
Since $g$ has integer coefficients and the $\strs_i$'s are bits, for $i=1,2\dotsc,p-1$,
\[
\strs_{i}=\strs_{i+p},
\]
and hence $\strs=\strt u\strt = \strt\circ u$, where $\strt$ consists of $p-1$ bits
and $u$ is a single bit.
Letting $a$ and $b$ denote the number of zeros and ones in $tu$, 
$P_\strs(x,y)=(1+x^ay^b)P_\strt(x,y)$. By theorem~\ref{thm:prm}
$P_\strt(x,y)$ has at most one non-palindromic factor. It suffices
to show that $1+x^ay^b$
which itself is palindromic has no non-palindromic factors.
If $a=0$ or $b=0$, $1+x^ay^b$ has no
non-palindromic factor. If both $a$ and $b$ 
are non-zero, since $a+b=p$, it has at most two factors one of which
by Lemma~\ref{lemma:sng_max_fct} must be $1+x$ or $1+y$ which is not possible
since $x=-1$ or $y=-1$ is not a root of $1+x^ay^b$.
\\[3mm]
\textit{Case 3:}
The proof follows as in Case 2 with modulo 2 operations,
\[
  g(x,x)=(1+x)(1+x+\dotsb+x^{p-1})=1+x^p \mbox{ (mod 2)}.\qedhere
\]
\end{proof}

\section{Lower bound on $\cnfmax_n$}
\label{S_lwr_bnd}
Recall that $\cnfmax_n$ is the largest number
of mutually-equicomposable $n$-bit strings.
In Section~\ref{S_smp_cnf} we saw that for all strings $\strs, \strt$,
the interleaved strings $\strs\circ\strt$ and $\strs\circ\strt^*$
are equicomposable. Choosing non-palindromic $\strs$ and $\strt$,
we derived confusable strings for all lengths $n$ such that
$n+1$ is a product of two integers $\ge 3$.
The interleaved strings are non-palindromic, hence $\cnfmax_n\ge 4$.
We now generalize the interleaved construction in two ways
to construct equicomposable sets that for some $n$ are as
large as $(n+1)^{\log_3 2}$.

We first show that if 
$\strs\sim\strs'$ and $\strt\sim\strt'$, then
\[
\strs\circ\strt\sim\strs'\circ\strt'.
\]

Any substring of $\strs\circ\strt$ is of the form
$
\textcolor{blue}{\strs_{\rm tail}}\,\strt_i\, \strs\, \strt_{i+1}\, \strs\,\dotsb\,
\strs\,\strt_{j}\,\textcolor{green}{\strs_{\rm head}},
$
where the \emph{tail} and \emph{head } substrings $s_{\rm
  tail}=s_{\ell}s_{\ell+1}\dotsb s_n$ and $s_{\rm head}=s_1s_2\dotsb
s_{h}$ can also be empty or all of $s$.
We show that it can be bijectively mapped to a substring 
$
\textcolor{blue}{\strs'_{\rm tail}}\,\strt'_{i'}\, \strs'\, \strt'_{i'+1}\, \strs'\,\dotsb\,
\strs'\,\strt'_{j'}\,\textcolor{green}{\strs'_{\rm head}}
$
of $s'\circ t'$, 
where $s'_{\rm   tail}=s'_{\ell'}s'_{\ell'+1}\dotsb s'_n$ and $s'_{\rm
  head}=s'_1s'_2\dotsb s'_{h'}$.

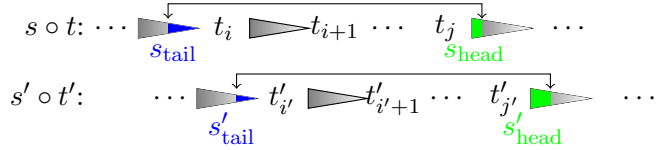
\begin{center}
\def\triangle{[shading=axis,shading angle=90]
      ++(0,0.2)-- ++(1.2,-0.2)-- ++(-1.2,-0.2)--cycle}
\def\L{2.25}

\begin{tabular}{r@{:}l}
$\strs\circ \strt$ &
\begin{tikzpicture}[baseline=-2.4pt,scale=0.65]
\node at (-0.5,0) {$\dotsb$};
\foreach \x/\y in {0/i,1/i+1,2/j,3}
{
\ifthenelse{\equal{\x}{2}}
{\node at (\L*\x+0.6,0) {$\dotsb$};}
{\draw (\L*\x,0) \triangle;}
\ifthenelse{\x<3}
{\node at (\L*\x+\L/2+0.6,0) {$t_{\y}$};}
{}
}
\node at (\L*3+2,0) {$\dotsb$};
\draw[<->] (0.6,0.2)--(0.6,0.5) --(\L*3+0.2,0.5)--(\L*3+0.2,0.2);
\node[blue] at (0.7,-0.5) {$s_{\rm tail}$};
\node[green] at (6.8,-0.5) {$s_{\rm head}$};
\begin{scope}
  \clip\triangle;
  \fill[blue] (0.6,-1) rectangle (2,1);
\end{scope}
\clip (\L*3,0) \triangle;
\fill[green!100] (\L*3,-1) rectangle (\L*3+0.2,1);
\end{tikzpicture}\\
$\strs'\circ \strt'$ &\qquad
\begin{tikzpicture}[baseline=-2.4pt,scale=0.65]
\node at (-0.5,0) {$\dotsb$};
\foreach \x/\y in {0/i',1/i'+1,2/j',3}
{
\ifthenelse{\equal{\x}{2}}
{\node at (\L*\x+0.6,0) {$\dotsb$};}
{\draw (\L*\x,0) \triangle;}
\ifthenelse{\x<3}
{\node at (\L*\x+\L/2+0.6,0) {$t_{\y}'$};}
{}
}
\node at (\L*4,0) {$\dotsb$};
\draw[<->] (0.8,0.2)--(0.8,0.5) --(\L*3+0.4,0.5)--(\L*3+0.4,0.2);
\node[blue] at (0.7,-0.6) {$s_{\rm tail}'$};
\node[green] at (6.8,-0.6) {$s_{\rm head}'$};
\begin{scope}
  \clip\triangle;
  \fill[blue] (0.8,-1) rectangle (2,1);
\end{scope}
\clip (\L*3,0) \triangle;
\fill[green!100] (\L*3,-1) rectangle (\L*3+0.4,1);
\end{tikzpicture}
\end{tabular}


\end{center}

Since $s\sim s'$, there is a bijection $f_{s,s'}$ that
maps every substring of $s$ to a substring of $s'$ with the same
composition, and a similar bijection $f_{t,t'}$.
Let $f_{t,t'}$ map the substring $t_i\dotsb t_j$ 
to a substring $t'_{i'}\dotsb t'_{j'}$ of the same composition, and
hence length.
Since $s\sim s'$, the strings $\strt_i\, \strs\, \strt_{i+1}\,
\strs\,\dotsb\,\strs\,\strt_{j}$
and
$\strt'_{i'}\, \strs'\, \strt'_{i'+1}\,
\strs'\,\dotsb\,\strs'\,\strt'_{j'}$
have the same composition.

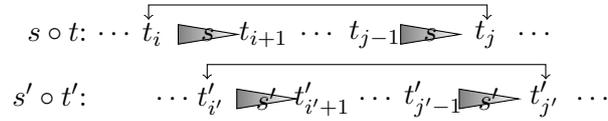
\begin{center}
\def\triangle{[shading=axis,shading angle=90]
      ++(0,0.2)-- ++(1.2,-0.2)-- ++(-1.2,-0.2)--cycle}
\def\L{2.25}

\begin{tabular}{r@{:}l}
$\strs\circ \strt$ &
\begin{tikzpicture}[baseline=-2.4pt,scale=0.65]
\node at (-\L+1,0) {$\dotsb$};
\node at (-\L/2+0.6,0) {$t_i$};
\foreach \x/\y in {0/i+1,1/j-1,2/j}
{
\ifthenelse{\equal{\x}{1}}
{\node at (\L*\x+0.6,0) {$\dotsb$};}
{\draw (\L*\x,0) \triangle;
 \node at (\L*\x+0.6,0) {$s$};
} 
\ifthenelse{\x<4}
{\node at (\L*\x+\L/2+0.6,0) {$t_{\y}$};}
{}
}
\node at (\L*3+0.5,0) {$\dotsb$};
\draw[<->] (-0.6,0.3)--(-0.6,0.6) --(\L*3-0.5,0.6)--(\L*3-0.5,0.3);
\end{tikzpicture}\\[1ex]
$\strs'\circ \strt'$ &\qquad
\begin{tikzpicture}[baseline=-2.4pt,scale=0.65]
\node at (-\L+1,0) {$\dotsb$};
\node at (-\L/2+0.6,0) {$t_{i'}'$};
\foreach \x/\y in {0/i'+1,1/j'-1,2/j'}
{
\ifthenelse{\equal{\x}{1}}
{\node at (\L*\x+0.6,0) {$\dotsb$};}
{\draw (\L*\x,0) \triangle;
 \node at (\L*\x+0.6,0) {$s'$};
}
\ifthenelse{\x<4}
{\node at (\L*\x+\L/2+0.6,0) {$t_{\y}'$};}
{}
}
\node at (\L*3+0.5,0) {$\dotsb$};
\draw[<->] (-0.6,0.4)--(-0.6,0.7) --(\L*3-0.5,0.7)--(\L*3-0.5,0.4);
\end{tikzpicture}
\end{tabular}


\end{center}

Using the tail and head notation above, it remains to show that $\ell$ and $h$ can be bijectively mapped to
$\ell'$ and $h'$ such that $\textcolor{blue}{\strs_{\rm  tail}}\textcolor{green}{\strs_{\rm head}}$
and $\textcolor{blue}{\strs'_{\rm    tail}}\textcolor{green}{\strs'_{\rm head}}$
have the same composition.
There are two cases. 

If the length of $s_{\rm tail}s_{\rm head}$ is at most $n$,
removing $s_{\rm tail}$ and $s_{\rm head}$ from $s$ yields
$s_{\rm middle}=s_{h+1}\dotsb s_{\ell-1}$, which $f_{s,s'}$ maps to a substring
$s'_{\rm middle}=s'_{h'+1}\dotsb s'_{\ell'-1}$ of $s'$.
Removing $s'_{\rm middle}$ from $s'$ yields a tail
$s'_{\rm tail}=s'_{\ell'}\dotsb s'_n$ and a head $s'_{\rm
  head}=s'_1\dotsb s'_{h'}$ of $s'$.
Their combined composition equals that of
the original head and tail.

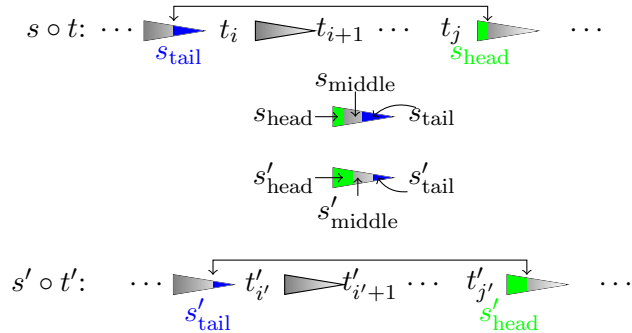
\begin{center}
\def\triangle{[shading=axis,shading angle=90]
      ++(0,0.2)-- ++(1.2,-0.2)-- ++(-1.2,-0.2)--cycle}
\def\L{2.25}

\begin{tabular}{r@{\hspace{-2ex}}c}
$\strs\circ \strt$: &
\begin{tikzpicture}[baseline=-2.4pt,scale=0.65]
\node at (-0.5,0) {$\dotsb$};
\foreach \x/\y in {0/i,1/i+1,2/j,3}
{
\ifthenelse{\equal{\x}{2}}
{\node at (\L*\x+0.6,0) {$\dotsb$};}
{\draw (\L*\x,0) \triangle;}
\ifthenelse{\x<3}
{\node at (\L*\x+\L/2+0.6,0) {$t_{\y}$};}
{}
}
\node at (\L*4,0) {$\dotsb$};
\draw[<->] (0.6,0.2)--(0.6,0.5) --(\L*3+0.2,0.5)--(\L*3+0.2,0.2);
\node[blue] at (0.7,-0.5) {$s_{\rm tail}$};
\node[green] at (6.8,-0.5) {$s_{\rm head}$};
\begin{scope}
  \clip\triangle;
  \fill[blue] (0.6,-1) rectangle (2,1);
\end{scope}
\clip (\L*3,0) \triangle;
\fill[green!100] (\L*3,-1) rectangle (\L*3+0.2,1);
\end{tikzpicture}\\
&
\begin{tikzpicture}[baseline=-2.4pt,scale=.65, inner sep=0]
\draw (0,0) \triangle;
\begin{scope}
\clip\triangle;
\fill[blue] (0.6,-1) rectangle (2,1);
\fill[green!100] (0,-1) rectangle (0.2,1);
\end{scope}
\node at (0.45,0.7) {$s_{\rm middle}$} edge[->] (0.45,0);
\node at (-1,0) {$s_{\rm head}$} edge[->] (0.1,0);
\node at (2,0) {$s_{\rm tail}$} edge[->,bend right] (0.8,0);
\end{tikzpicture}\\[2ex]
&\begin{tikzpicture}[baseline=-2.4pt,scale=.65, inner sep=0]
\draw (0,0) \triangle;
\begin{scope}
\clip (0,0) \triangle;
\fill[blue] (0.8,-1) rectangle (1.2,1);
\fill[green!100] (0,-1) rectangle (0.4,1);
\end{scope}
\node at (0.5,-0.7) {$s_{\rm middle}'$} edge[->](0.5,0);
\node at (-1,0) {$s_{\rm head}'$} edge[->] (0.2,0);
\node at (2,0) {$s_{\rm tail}'$} edge[->,bend left] (0.9,0);
\end{tikzpicture}\\[5ex]
$\strs'\circ \strt'$: &\qquad
\begin{tikzpicture}[baseline=-2.4pt,scale=0.65]
\node at (-0.5,0) {$\dotsb$};
\foreach \x/\y in {0/i',1/i'+1,2/j',3}
{
\ifthenelse{\equal{\x}{2}}
{\node at (\L*\x+0.6,0) {$\dotsb$};}
{\draw (\L*\x,0) \triangle;}
\ifthenelse{\x<3}
{\node at (\L*\x+\L/2+0.6,0) {$t_{\y}'$};}
{}
}
\node at (\L*4,0) {$\dotsb$};
\draw[<->] (0.8,0.2)--(0.8,0.5) --(\L*3+0.4,0.5)--(\L*3+0.4,0.2);
\node[blue] at (0.7,-0.6) {$s_{\rm tail}'$};
\node[green] at (6.8,-0.6) {$s_{\rm head}'$};
\begin{scope}
  \clip\triangle;
  \fill[blue] (0.8,-1) rectangle (2,1);
\end{scope}
\clip (\L*3,0) \triangle;
\fill[green!100] (\L*3,-1) rectangle (\L*3+0.4,1);
\end{tikzpicture}
\end{tabular}


\end{center}

Similarly, if the length of $s_{\rm tail}s_{\rm head}$ is greater than $n$,
$s_{\rm tail}$ and $s_{\rm head}$ have a common substring
$s_{\rm middle}=
s_{\ell}\dotsb s_{h}$,
which $f_{s,s'}$ maps to a substring $s_{\rm middle}'=s'_{\ell'}\dotsb s_{h'}$,
in $s'$.
The head $s'_{\rm head}=s'_1\dotsb s'_{h'}$ and tail $s'_{\ell'}\dotsb
s'_{n}$ of $s'$ 
have the same combined composition as the original head and tail.

The second generalization of the interleaved construction is to interleaving more than two strings.
The following are some properties of interleaving strings.\\[3mm]
\textit{Associativity:}
$(\strs_1 \circ \strs_2) \circ \strs_3
=\strs_1 \circ (\strs_2\circ \strs_3).
$\\[3mm]
\textit{Unique factorization:}
Every string has a unique maximal factorization into interleaved strings.
\begin{proof} 

We first show that if $s\circ s'=t\circ t'$ where $s$ and $t$
are irreducible under the interleaving operation, then both equal
$u\circ u'$ for some string $u$ of length $(|s+1|,|t+1|)-1$.
If $|s|=|t|$ then we easily get $s=t$, else from lengths we have
$|s+1|\cdot|s'+1|=|t+1|\cdot|t'+1|$ and hence  $(|s+1|,|t+1|)=|u+1|>1.$
We can expand $s
=u_1 a_1 u_2 a_2\cdots u_M $, where $M=|s+1|/|u+1|$ and
$t
=v_1 b_1 v_2 b_2\cdots v_N $, where $N=|t+1|/|u+1|$
where $|u_i|=|v_i|=|u|$ and each $a_i$ and $b_i$ are bits. Then 
since $(M, N)=1$ for any $(i, j) \in ([M],[N])$ there are substrings
$s_{1}$ and $t_1$
of $s\circ s'$ and $t\circ t'$ starting and ending at same index
such that $s_1=u_i$ and $t_1=v_j$. This implies that $u_i=v_j$
for all $i$ and $j$
and thus each of $s$ and $t$ are reducible.
\end{proof}
\begin{thm}
\label{theorem:cnf_of_intr}
\[
\cnfset_{\strs_1\circ\strs_2\circ\dotsb\circ\strs_k}
=
\sets{\strt_1\circ\strt_2\circ\dotsb\circ\strt_k:
\strt_i\sim\strs_i\textit{ for all }i
}.
\]
\end{thm}
\begin{proof}
Proving $\supseteq$ is simple.
By associativity and induction,
if $\strs_i\sim\strt_i$  for $i=1\upto m$, then
\[
\strs_1 \circ \strs_2 \circ \dotsb \circ \strs_m
\sim
\strt_1 \circ \strt_2 \circ \dotsb \circ \strt_m.
\]
The proof of $\subseteq$ is more complex and 
postponed to Appendix~\ref{sec:appendixA}.
\end{proof}

The simple ($\supseteq$) part of the theorem suggests a simple construction of large confusable sets.
Every string is equicomposable with its reversal, hence
\[
\cnfset_{\strs_1\circ\strs_2\circ\dotsb\circ\strs_k}
\supseteq
\sets{\tilde\strs_1\circ\tilde\strs_2\circ\dotsb\circ\tilde\strs_k:
\tilde\strs_i\textit{ is either }\strs_i\textit{ or }\strs_i^*}.
\]
If all these strings are non-palindromic, then by unique factorization,
each of the resulting interleaved products is different, hence
\[
|\cnfset_{\strs_1\circ\strs_2\circ\dotsb\circ\strs_k}|
\ge
2^k.
\]
Let $|\strs_i|=m_i$. We saw that
$|\strs_1\circ\strs_2|=(m_1+1)(m_2+1)-1$, and by induction
\[
|\strs_1\circ\strs_2\circ\dotsb\circ\strs_k|=\prod_{i=1}^k(m_i+1)-1.
\]
For example, taking $s_1=\ldots=s_k= 01$ and $n=3^k-1$,
\[
\cnfset_n
=
\cnfset_{3^k-1}
\ge
|\cnfset_{\underbrace{\scriptstyle{01\circ01\circ\dotsb\circ01}}_k}|
\ge
2^k
=
(n+1)^{\log_3 2}.
\]
For general lengths $n$, consider the distinct-prime factorization
\[
n+1=2^{e_0}p_1^{e_1}p_2^{e_2}\ldots p_k^{e_k}.
\]
For every prime $p\ge 3$, take a non-palindromic string of
length $p-1$,
and for every pair of $2$'s, take a non-palindromic string of
length $2\cdot2-1=3$.
Interleaving all these strings and their reversals lower bounds
the largest number of mutually equicomposable strings.
\begin{thm}
\label{theorem:lwr_bnd_cnf}
\[
\cnfmax_n
\ge 
2^{\Floor{\frac{e_0}{2}}+e_1+e_2+\cdots+e_k}. 
\]
\end{thm}

As mentioned above, this is the strongest lower bound we have. In the rest
of the section we discuss the structure of possible extensions.

The confusable strings we have seen so far were interleaved as in
Theorem~\ref{theorem:cnf_of_intr}.
Exhaustive search shows that so are all confusable strings of length
$\le 22$.
Yet Theorem~\ref{theorem:cnf_of_intr} shows that equicomposable sets larger than implied by
Theorem~\ref{theorem:lwr_bnd_cnf} must be based on non-interleaved
strings. 
We next show examples of non-interleaved confusable strings.
However, while showing that non-interleaved confusions may arise,
these examples do not necessarily result in larger confusable sets.

We begin with a polynomial interpretation of interleaving.
\begin{lem}
\label{lem:polycomp}
For any $\strs$ with $a$ 0's and $b$ 1's and any $\strt$,
\[
P_{\strs\,\circ\,\strt}(x,y)=P_\strs(x,y)P_\strt(x^{a+1}y^b,x^ay^{b+1}).
\]
\end{lem}
\begin{proof}
Let $\strt\ed t_1t_2\dotsb t_m$, so that $\strs\circ \strt=\strs t_1\strs
t_2\dotsb t_m \strs$,
and define 
$a_i=\sum_{j=1}^it_j$ and $b_i=\sum_{j=1}^i(1-t_j)$
to be the numbers of 0's and 1's in $t_1t_2\dotsb t_i$.
Then,
\begin{align*}
P_{\strs\,\circ\,\strt}(x,y)
&=P_\strs(x,y)+ x^{a+t_1}y^{b+(1-t_1)}P_{\strs}(x,y)
+\dotsb
+ x^{ma+t_1+\dotsb +t_m}y^{mb+(1-t_1)+\dotsb+(1-t_m)}P_{\strs}(x,y)\\
&=P_\strs(x,y)\sum_{i=1}^m x^{ia+a_i}y^{ib+b_i}.
\end{align*}
The lemma follows as
\[
\sum_{i=1}^m x^{ia+a_i}y^{ib+b_i}
\!=\!\sum_{i=1}^m x^{(a_i+b_i)a+a_i}y^{(a_i+b_i)b+b_i}
\!=\!\sum_{i=1}^m (x^{a+1}y^b)^{a_i}(x^ay^{b+1})^{b_i}
\!=\!P_{\strt}(x^{a+1}y^b,x^ay^{b+1}).\qedhere
\]
\end{proof}
Following are more general classes of equicomposable strings.

The following result follows easily from Theorem~\ref{thm:plycmp}. This 
is the most general class of strings which are equicomposable with each other
we have come up with.
\begin{thm}
\label{thm:maxweknow}
Let $\strs_1,\strs_2,\dotsc,\strs_k$ be strings whose generating polynomials
have a common factor $D(x,y)$.
Derive $\strt_1,\strt_2,\dotsc,\strt_k$ from $\strs_1,\strs_2,\dotsc,\strs_k$ by replacing $D(x,y)$ by its
reciprocal.
Then for any $x_1x_2\dotsc x_{k-1}$,
\begin{align}
\strs_1\,x_1\,\strs_2\,x_2\,\dotsb\, x_{k-1}\,\strs_k
\sim
\strt_1\,x_1\,\strt_2\,x_2\,\dotsb\, x_{k-1}\,\strt_k\label{str:T}.
\end{align}
\end{thm}
The equicomposable strings mentioned earlier are a
special case of the theorem.
The following corollaries yield classes of equicomposable strings with
some structure.
\begin{cor}
\label{crlcls}
Let $A\sim B$ and for $i=1,2,\dotsc,k$ let $(\strs_i,\strt_i)\in\{(A,B'),(B,A')\}$.
Then, for any string $\bar x=x_1x_2\dotsb x_{k-1}$,
the two strings 
in \eqref{str:T} are equicomposable.\hfill\qedsymbol
\end{cor}

\begin{cor}
\label{crlcnf}
Let $\strs_1$, $\ldots$, $\strs_k$ be strings of the same composition.
Then for any string $\strs_0$ and bits $x_1x_2\ldots x_{k-1}$,
\[
(\strs_1\circ \strs_0 )x_1 (\strs_2\circ \strs_0)x_2 \ldots x_{k-1}( \strs_k\circ \strs_0)\sim
(\strs_1\circ \strs_0^*) x_1 (\strs_2\circ \strs_0^*)x_2\ldots x_{k-1}
(\strs_k\circ \strs_0^*).\eqno\qedsymbol
\]
\end{cor}

For example, taking $\strs_1=010$, $\strs_2=001$, $\strs_0=01$, and $x_1=0$, we
obtain the 23-bit confusion
\begin{align*}
\textbf{010}0\textbf{010}1\textbf{010}\textit{0}\,\textbf{001}0\textbf{001}1\textbf{001}
&=
(\textbf{010}\circ01)\textit{0}\,(\textbf{001}\circ01)\\
&\sim
(\textbf{010}\circ10)\textit{0}\,(\textbf{001}\circ10)\\
&=
\textbf{010}1\textbf{010}0\textbf{010}\textit{0}\,\textbf{001}1\textbf{001}0\textbf{001}.
\end{align*}
These strings are clearly non-interleaved, hence are the shortest
non-interleaved confusable strings.

This example shows in particular that string equicomposability differs
from partition coarsening where~\cite{BTW06} used Ribbon Schur Functions
to show that all confusions arise from interleaving.

\section{Upper bounds on $\cnfmax_n$}
\label{S_upr_bnd}
To upper bound the largest number of mutually-equicomposable strings,
we use well-known results on the factorization of $x^n-1$.

\begin{thm}
For all $n$,
\label{theorem:upr_bnd_cnf}
\[
\cnfmax_n
\le \min\left\{2^{d(n+1)-1}, (n+1)^{1.23}\right\}.
\]
\end{thm}
\begin{proof}
The $(n+1)^{1.23}$ term follows from results in~\cite{LW88}
by replacing $y$ in $P(x,y)$ by $x^{n+1}$ to obtain a univariate
polynomial, and reproving theorems therein for this
polynomial, showing that $\cnfmax_n\le (n+1)^{1.23}$.

To prove the first upper bound, let $\strs$ have length $n$, 
then 
\[
P_\strs(x,x)=1+x+x^2+x^3+\ldots+x^n
=
\frac{x^{n+1}-1}{x-1}.
\]
Hence the prime factorization of $P_\strs(x,x)$ contains $d(n+1)-1$ terms.
From Lemma~\ref{lemma:fct_xy_fct_xx}
the prime factorization of $P_\strs(x,y)$ contains at most that many
terms,
and the bound follows from Corollary~\ref{cor:plycmp}.
\end{proof}
To also relate the form of this upper bound to that of the lower bound in Theorem~\ref{theorem:lwr_bnd_cnf},
again let $n+1=2^{e_0}p^{e_1}_1p_2^{e_2}\ldots p_k^{e_k}$ be a
factorization into distinct primes.
Then $d(n+1)=(e_0+1)(e_1+1)\cdots(e_k+1)$, hence
\[
\cnfmax_n
\le \min\left\{2^{(e_0+1)(e_1+1)\cdots(e_k+1)-1}, (n+1)^{1.23}\right\}.
\]

The lower bound in Theorem~\ref{theorem:lwr_bnd_cnf}
is tight when $n+1$ is a prime power or twice a prime power.
\begin{thm}
\label{thmprmpwr}
For any $k$,
\[
\cnfmax_{2^k-1}= 2^{\Floor{\frac{k}{2}}},
\]
and for prime $p\ge3$,
\[
\cnfmax_{p^k-1}=
2^k.
\]
\end{thm}
\begin{proof}
Eisenstein's criterion again implies the prime factorization
\begin{align*}
P_\strs(x,x)
&=
1+x+x^2+\dotsb+x^{2^k-1}= \prod\limits_{j=1}^{k-1}(1+x^{2^j}).
\end{align*}
Once again, to prove the bound it suffices to show that any factorization
of $P_s(x,y)$ into 
$f(x,y)g(x,y)$ with $g(x,x)=1+x^{2^i}$ implies that
$g(x,y)$ is palindromic. 

Let $\strs=\strs_1\strs_2\ldots\strs_{2^k-1}$ and as before, let $a_i$ be the number of
zeros among $s_1\upto s_i$.
Then
$P(x,y)\ed P_\strs(x,y)=\sum_{i=0}^{2^k-1}x^{a_i}y^{i-a_i}$.
Denote 
$\pd[P]{x}(x,y)|_{y=x}$ by $ P_x'$, then 
\begin{align*}
(1-x)P_x'&=
(1-x)f(x,x)g_x'+(1-x)g(x,x)f_x'\\
\implies\sum_{i=0}^{2^k-2}s_{i+1}x^{i}-a_{2^k-1}x^{2^{k}-1}&
=\left(1-x^{2^i}+x^{2\cdot2^i}-x^{3\cdot2^i}+
\dotsb-x^{(2^{k-i}-1)\cdot2^i}\right)g_x'
+(1+x^{2^i})(1-x)f_x'.
\end{align*}
The degree of $g_x'$ is $2^i-1$, and that of 
$(1-x)f_x'$ is $2^k-1-2^i$, so we can write
\[
(1-x)f_x'=\sum_{j=0}^{2^{k-i}-2}x^{j\cdot2^i} f_j,
\]
where each $f_j$ has degree $\le 2^i-1$.
Defining $f_{-1}=f_{2^k}=0$ and $s_{2^k}=-a_{2^k-1}$, 
\begin{align*}
\sum_{i=0}^{2^k-2}s_{i+1}x^{i}-a_{2^k-1}x^{2^{k}-1}
&=\sum_{j=0}^{2^{k-i}-1}(-1)^jx^{j\cdot2^i}g_x'+\sum_{j=0}^{2^{k-i}-2}x^{j\cdot2^i}
f_j+\sum_{j=1}^{2^{k-i}-1}x^{j\cdot2^i} f_j\\
\implies \sum_{j=0}^{2^{k-i}-1}x^{j\cdot2^i}
\left(\sum_{l=0}^{2^{i}-1}s_{j\cdot2^i+l+1}x^{l}\right)
&= \sum_{j=0}^{2^{k-i}-1}x^{j\cdot2^i}\left((-1)^jg_x'+
f_{j-1}+f_j\right)
\end{align*}

In the above equation, if we sum all the terms with even
values of $j$ and subtract from the sum of odd values
we get, 
\begin{align*}
\sum_{j=0}^{2^{k-i}-1}(-1)^{j\cdot2^i}\left(\sum_{l=0}^{2^{i}-1}s_{j\cdot2^i+l+1}x^{l}\right) =2^{k-i}g_x'.
\end{align*}
Since all $s_l$'s except $s_{2^k}$ are either 0 or 1, we must have
\[
g_x'= a\cdot x^{2^i-1},
\]
and thus, 
\[
g(x,y)=1+g_1(x,y)+x^ay^{2^i-a}.
\]
\begin{clm}
If $f(x,x)=K$, a constant and $f_x'=0$, then $(x-y)^2|f(x,y)-K$.
\end{clm}
\begin{proof}
Let $f(x,y)=(x-y)f_1(x,y)+f_2(x)$. Then, clearly $f_2(x)=K$.
Also, $f_x'=0$ implies $f_1(x,x)=0$
and thus $f_1(x,y)$ has $x-y$ as a factor. 
\end{proof} 
Using the claim we conclude that $g(x,y) = 1+(x-y)^2g_1(x,y)+x^ay^{2^i-a}$.
Consider, 
\[
  f(x,y)\left(g-g^*\right).
\]
It is the difference of two generating polynomials, but substituting
$y=-x$, we see that each coefficient is divisible by 4 which cannot
be true unless $g-g^*=0$.

It is easy to see that the bound is achievable by taking
strings of the form $t_0\circ t_1\circ t_2\circ \cdots\circ
t_\Floor{\frac{k}{2}}$
where $t_i=001$ for $i\ge1$, and $t_0$ is the empty string when $k$ is
even, else it is the single bit 0. 

When $n=p^k-1$ for a prime $p\ge 3$, Theorems~\ref{theorem:lwr_bnd_cnf}
and~\ref{theorem:upr_bnd_cnf} coincide.
\end{proof}

\section{Reconstruction algorithm}
We present an algorithm for reconstructing strings from their
composition multiset.
The algorithm takes as input the composition multiset $\cS_{\strs}$
of a string $\strs$ over an alphabet $\ab$, and elements of 
$\cnfset_{\strs}$,  the set of all strings confusable with $\strs$.
We show that for alphabet size $\ge4$, the string $\strs$ that 
generated $\cS_\strs$ is added in quadratic time. 
The algorithm successively reconstructs $\strs$ from both ends
and backtracks when it errs. It can be viewed as a modification
of a similar algorithm for the turnpike problem.

We first establish two properties of $\cS_{\strs}$
that help reduce the algorithm's search space.
The next lemma shows that the composition multiset 
determines the set $\{s_i,s_{n+1-i}\}$
of symbols at the symmetric positions $i$ and $(n+1-i)$
for $i = 1,2,\ldots,\lceil\frac{n}2\rceil$.
\begin{lem}
\label{lemmirror}
 $\cS_\strs$
determines $\{s_i, s_{n+1-i}\}$ for $i = 1,2,\ldots,\lceil\frac{n}2\rceil$.
\end{lem}
\begin{proof}
Let the union of compositions be their union as multisets.
For example, $A^2B\cup ABC^2\cup AC=A^4B^2C^3$. 
For a string $\strs$, let $M_i$ denote the union of the compositions of
all substrings of length $i$.
For example, for ABAC, $M_1=A^2BC$, $M_2= A^3B^2C$.
Note that all $M_i$'s can be easily determined from the string, and
that for $1\le i\le\floor{n/2}$, $M_{n+1-i}=M_i$.
For a multiset $S$, let $j\cdot S$ be the $j$-fold union
$S\union\ldots\union S$. 
It is easy to see that
\[
M_2\cup\sets{s_1,s_n}=2\cdot M_1,
\]
hence $\sets{s_1,s_n}$, can be deduced from $\cS_\strs$.
More generally, for $i=2\upto\ceil{\frac n2}$,
\begin{align*}
M_i\cup\sets{s_{i-1}, s_{n-i+2}}\cup
2\cdot\sets{s_{i-2}, s_{n-i+3}}\cup\ldots\cup (i-1)\cdot\sets{s_1, s_n}
=
i\cdot M_1.
\end{align*}
Using this equation inductively over $i=2\upto \ceil{\frac n2}$,
yields all multisets $\{s_i,s_{n-i+1}\}$.
\end{proof}

For $1\le i<\ceil{n/2}$, let $\cT_i$ be the collection of compositions of
strings $\strs_j^k$ where $j,k\le i$, or $j,k\ge n+1-i$, or $j\le i+1\le
n-i\le k$, namely the collection of compositions of all strings that
are either on ``one side'' of $\strs_{i+1}^{n-i}$ or ``straddle''
$\strs_{i+1}^{n-i}$.
The next lemma shows that the composition of the whole string, 
along with the strings $\strs^i$ and $\strs_{n+1-i}^n$ determine $\cT_{i}$.
\begin{lem}
\label{lemcommult}
$\cS_\strs$, $\strs^i$, and $\strs_{n+1-i}^n$ determine $\cT_i$.
\end{lem}
\begin{proof}
$\cT_i$ consists of compositions of three types of 
strings, those that are substrings of $s_1^i$, that are
substrings of  $s_{n+1-i}^n$, and substrings that cover 
all symbols in between, \ie $s_{i+1}^{n-i}$.
The first and last $i$ symbols determine the compositions
of the first two type of strings. The third type of strings
are those that contain the string $s_{i+1}^{n-i}$. This 
is a \emph{symmetric substring}, namely it has the same
number of symbols on its left and right, 
and by Lemma~\ref{lemmirror} we can determine its composition.
Knowing this composition and the first and last $i$ symbols 
yields the multiset of strings of the third type.
\end{proof}






We use the two lemmas to reconstruct the string.
Let $w(s)$ denote the composition of a string $s$.
\begin{lem}
\label{lemcnf}
If $w(s_1^{i})\ne w(s_{n-i+1}^{n})$, then $\cS_{\strs}$,  $s_1^i$, and
 $s_{n+1-i}^n$  determine the pair $(s_{i+1}, s_{n-i})$.
\end{lem}
\begin{proof}
By Lemma~\ref{lemcommult}, we can determine $\cT_i$.
Consider the two longest compositions in $\cS_{\strs}\setminus
\cT_{i}$. They correspond to the length-$(n-i-1)$ strings $s_1^{n-i-1}$ and
$s_{i+2}^{n}$.
The complements of these two compositions are therefore the
compositions of $s^{i+1}$ and $s_{n-i}^n$.

By Lemma~\ref{lemmirror}, we can also derive the multiset
$\sets{s_{i+1}, s_{n-i}}$.
If $s_{i+1}=s_{n-i}$, then we can determine $s^{i+1}$ and $s_{n-i}^n$.
Otherwise, $s_{i+1}\ne s_{n-i}$, and since $w(s_1^i) \ne
w(s_{n+1-i}^n)$, it is easy to see that 
$\sets{w(s_1^i)\cup
w(s_{i+1}),  w(s_{n-i}) \cup w(s_{n-i+1}^n)}\ne
\sets{w(s_1^i)\cup
w(s_{n-i}),  w(s_{i+1}) \cup w(s_{n-i+1}^n)}$, 
hence we can determine the pair $(s_{i+1}, s_{n-i})$.
\end{proof}



The algorithm reconstructs the string by sequentially deciding 
on the values of the pair of symbols $s_i$ and $s_{n-i+1}$, that are 
in symmetric positions.
It starts with the symbols $s_1$ and $s_n$ at the ends of the string,
and progressively moves closer to the center.
The algorithm first determines $s_1s_2$ and $s_{n-1}s_n$, uniquely
upto reversal. 
By Lemma~\ref{lemmirror}, we know the
multiset $\sets{s_1,s_n}$ and
can arbitrarily decide which is $s_1$ and which is $s_n$.
It next determines $s_2$ and $s_{n-1}$. Again by the lemma
we know $\sets{s_2,s_{n-1}}$, and if $s_1=s_n$ we can
decide on $s_2$ and $s_{n-1}$ arbitrarily, while if
$s_1\ne s_n$, by Lemma~\ref{lemcnf}, we can determine
$s_2$ and $s_{n-1}$. 
For $\sets{s_3,s_{n-2}}$, 
the ends $s_1^2$ and $s_{n-1}^n$ can differ, while
their weights could be the same. In such cases Lemma~\ref{lemcnf}
cannot be applied. 
 However if $s_3=s_{n-2}$, which can be determined
from the Lemma~\ref{lemmirror}, we easily determine $(s_3,s_{n-2})$. 
In other words, if $s_i=s_{n+1-i}$, the algorithm easily determines
 $(s_i,s_{n+1-i})$.
Therefore, from this point on, when the algorithm has reconstructed
the first and last $i$ symbols, 
and $w(s_1^{i})= w(s_{n+1-i}^{n})$ but $s_{i+1}\ne s_{n+1-i}$,  it
guesses one of the two possibilities for $(s_{i+1},s_{n+1-i})$ and 
proceeds, while keeping track of the number $t$ and locations
$i_1<i_2<\ldots<i_t$ of locations where guesses were made. 
After determining/guessing $s_{i+1}$ and $s_{n-i}$, the algorithm
updates $\cT_{i}$ to $T_{i+1}$, which can be accomplished in near-linear time. It then checks
whether $\cT_{i+1}\subseteq\cS_{\strs}$ as multisets. 
If at some
point $\cT_{i+1}\subsetneq\cS_{\strs}$, the algorithm \emph{backtracks}.
It changes its guess at location $i_t$ by swapping 
$s_{i_t}$ and $s_{n+1-i_t}$, changes $t$ to $t-1$,
and restarts reconstruction from location $i_t+1$.
The process continues until the whole string is reconstructed,
namely $i=\ceil{\frac n2}$ and $\cT_{\ceil{n/2}}=\cS_{\strs}$.
If at that point there are $t\ge 1$ guesses, then as before
the algorithm backtracks to guess $t$ and tries to find
additional strings in $\cnfset_{\strs}$.


Since our algorithm relies on the compositions of strings from both 
ends, it is helpful to define a few terms. Let
 \[
\numwt_\strs\ed\left| \{i<n/2 : w(\strs^i)=w(\strs_{n+1-i}^n) \text{
    and } s_{i+1}\ne s_{n-i}\}\right|
\] 
be the number of substrings from the ends having the same composition,
and the next two symbols are distinct. 
Let 
\[
\nummax_\strs \ed\max_{\strt\in\cnfset_\strs}\numwt_{\strt}
\]
be the largest value of $\numwt$ over all strings in $\cnfset_\strs$.

The backtracking algorithm induces a binary tree where the
nodes represent locations
at which there are two possible reconstructions. 
The procedure described above does a depth-first search.
Instead we could also do a breadth first search, where 
we consider all branches at once, namely at any time, 
all potential reconstructions correspond to level $t$
or $t+1$. This implies that given $\cS_{\strs}$, the 
algorithm is able to find $\strs$ before depth $\numwt_\strs+1$.

We bound the time required to reconstruct $\strs$ from $\cS_\strs$
in the following theorem.
\begin{thm}
\label{thm:main_reconstruction}
The backtracking algorithm, run using proper data structures,
given $\cS_\strs$ and $\numwt_\strs$, 
outputs a set of strings that contains $\strs$ in time
$O(2^{\numwt_{\strs}}n^2\log n)$. Futhermore, $\cnfset_\strs$
can be recovered in time $O(2^{\nummax_{\strs}}n^2\log n)$.
\end{thm}
\begin{proof}

We assume an arbitrary order over the elements of
$\Sigma$. This introduces a lexicographical ordering
over compositions of strings on $\Sigma$. We use a
Red-Black tree~\cite{CormenLRS01} to store multisets
of compositions.
The advantage of this data structure is that insertion, deletion
and search all require time $O(\log n)$, where $n$ is the 
size of the data.

Notice that $\cT_{i+1}$ is obtained from $\cT_i$ by adding 
the compositions of substrings that have an endpoint at $s_{i+1}$
or $s_{n-i}$. In particular, at most $2n$ compositions are 
added, requiring $O(n\log n)$ time. For each branch, we keep
a copy of $\cS_\strs$ and we prune it to populate $\cT$, \ie
while constructing $\cT_{i+1}$ we simultaneously remove the
new entries/compositions from the copy of $\cS$ corresponding 
to this branch, requiring another $O(n\log n)$. When there
are two possibilities for reconstruction at any stage, 
we make copies of $\cT_i$ and $\cS$ corresponding to 
that node and proceed along each. This step takes time
$O(n^2\log n)$. 

The algorithm while reconstructing $\strs$ does not 
``fork out'' more then $\numwt_\strs+1$ times, therefore
the number of branches before reconstructing $\strs$ is
at most $2^{\numwt_\strs}$. Along each path we make $n/2$
iterations requiring a total time of $O(n^2\log n)$. 
Combining these, the total complexity of reconstructing
$\strs$ is at most $O(2^{\numwt_\strs}n^2\log n)$.

To reconstruct $\cnfset_\strs$, we note that the number of
``forks'' is $\le \nummax_\strs$. Therefore, a similar 
computation shows that the algorithm runs in time
$O(2^{\nummax_\strs}n^2\log n)$.
\end{proof}

We bounded the run time of reconstructing a string
as a function of $\numwt_\strs$. For random strings, 
we bound the run time by studying the distribution 
of $\numwt_\strs$. 

We do this by using some well known results in random walk theory.
These are applications of the Striling approximation of factorials.
\begin{lem}[Stirling's approximation]
For any $n\ge1$, there is a $\theta_n\in(\frac1{12n+1},\frac1{12n})$
such that 
\[
n! = \sqrt{2\pi n}\left(\frac n e \right)^n e^{\theta_n}.
\]
\end{lem}

We now apply this to bound the 
probability that two random strings have the same 
composition. A stronger version of the result that finds asymptotic
equality is 
proved in~\cite{RS09}, but is not required for 
our purposes. 
\begin{lem}
\label{lem:rnd}
For $|\Sigma|= k$, let $\strs$ and $\strt$ be two 
random  length-$n$ strings over $\Sigma$, then for $k\ge n$
\[
 P\left(w(\strs)=w(\strt)\right)< \frac{n!}{k^n},
\]
and for $k<n$
\[
P\left(w(\strs)=w(\strt)\right)<\frac{k^{k/2}e^{1/12n}}{(2\pi n)^{(k-1)/2}}.
\]
\end{lem}
\begin{proof}
The probability that the symbols in $\Sigma$ appear
$i_1\upto i_k$ times respectively in a random length-$n$ string is
\[
\frac1{k^{n}}{n \choose i_1\upto i_k}.
\]
Therefore, the probability that two random strings have the same
composition is
\begin{align*}
\sum_{i_1+\ldots+ i_k=n}\frac1{k^{2n}}{n \choose i_1\upto i_k}^2
\le &\frac1{k^{n}}\max_{i_1+\ldots+ i_k=n}{n \choose i_1\ldots i_k}\cdot\sum_{i_1+\ldots+ i_k=n}\frac1{k^{n}}{n \choose i_1\upto i_k}\\
= & \frac1{k^{n}}\max_{i_1\upto i_k}{n \choose i_1\ldots i_k},
\end{align*}
where the last step follows from
\[
 \sum_{i_1+\ldots+ i_k=n}\frac1{k^{n}}{n \choose i_1\upto i_k} = 1.
\]
For $k\ge n$,
\[
\max_{i_1\upto i_k}{n \choose i_1\ldots i_k}= n!
\]
implying the first part.

For $n\ge k$, note that
\[
f(x)\ed\sqrt{2\pi x}\left(\frac{x}{e}\right)^x
\]
is convex in $(1,\infty)$. 
Therefore,
\[
 \prod_{j=1}^k i_j! \stackrel{(a)}{>}  \prod_{j=1}^k f(i_j) \stackrel{(b)}{>} \left(f\left(\frac n k\right)\right)^k
 = \left(\sqrt{2\pi \frac nk}\right)^k\left(\frac{n}{ke}\right)^n.
\]
where $(a)$ uses Stirling's approximation, and $(b)$
follows from convexity of $f$. Hence
\begin{align*}
&\max_{i_1\upto i_k}{n \choose i_1\ldots i_k}
< \frac{n!}{\left(f\left(\frac n k\right)\right)^k}
< \frac{k^{k/2}e^{1/12n}}{(2\pi n)^{(k-1)/2}},
\end{align*}
where in the last step we approximate $n!$.
\end{proof}

We now study the distribution of $\numwt_\strs$.
We first consider alphabet size $\ge4$ in detail,
and then provide performance guarantees for 
alphabet size $3$ and $2$. 

\section{Alphabet size $\ge4$}

Consider two uniformly random independent infinite strings $\strs^\infty$ 
and $\strt^{\infty}$ over $\Sigma$. 
A set of integers $i_1<i_2<\ldots<i_m$ is \emph{non-consecutive} if $i_{j+1}>i_j+1$.
Let $F_m$ be the event that there are 
at least $m+1$ \emph{non-consecutive} integers $i_0\ed 0<i_1<i_2<\ldots< i_m$ 
such that for each $j\le m$, $w(\strs_1^{i_j})=w(\strt_1^{i_j})$. 
After a location $i_j$ at which $w(\strs_1^{i_j})=w(\strt_1^{i_j})$
by independence the process is equivalent to starting at 
time 0. It follows that $P(F_{j+1}|F_j) =P(F_1)$. 
Therefore, 
\begin{align}
\label{eqn:geom}
P(F_m) = P(F_1)^m.
\end{align}
Let $M(\strs^{\infty},\strt^{\infty})$ denote the total number of 
non-consecutive integers for which $w(\strs_1^{i})=w(\strt_1^{i})$. 
Then by Equation~\eqref{eqn:geom}, 
\[
\EE[M] = \sum_{m\ge1} P(M\ge m) = \sum_{m\ge1} P(F_m) = \frac{P(F_1)}{1-P(F_1)}.
\]
However, if instead of non-consecutiveness, we only restrict $i_1\ge2$, 
Lemma~\ref{lem:rnd} shows that 
\[
\EE[M] \le \sum_{n=2}^k\frac{n!}{k^n} +\sum_{n=k+1}^\infty\frac{k^{k/2}e^{1/12n}}{(2\pi n)^{(k-1)/2}}.
\] 
The right hand side of this equation is finite for $k\ge4$. In fact, it decays 
as $O(1/k^2)$ with $k$. 
This implies that $p_k\ed P(F_1) <1$ and therefore for a random string
$\strs$, Equation~\eqref{eqn:geom} gives
\[
P(\numwt_\strs>m) \le P(F_m) = p_k^m,
\]
proving the following lemma.
\begin{lem}
A random string over alphabet size $k\ge4$, with probability $\ge
1-p_k^m$, satisfies $\numwt_\strs\le m$. 
\end{lem}
This implies that with probability $\ge 1-\delta$, 
$\numwt_\strs<\log_{1/p_k} \frac1\delta$.
Therefore applying Theorem~\ref{thm:main_reconstruction}, 
\begin{thm}
\label{thm:four}
For a random string $\strs$ over an alphabet of size $\ge4$ the
backtracking algorithm with probability $>1-\delta$,
outputs a subset of $\cnfset_s$ containing $\strs$ in time
$O_{\delta,k}(n^2\log n)$.  
\end{thm}

Recall from Theorem~\ref{theorem:upr_bnd_cnf} that
$|\cnfset_\strs|<n^{1.23}$ and therefore by the union bound,
\begin{lem}
With probability $\ge 1-n^{1.23}p_k^m$, a random string over
alphabet size $k\ge4$ satisfies $\nummax_\strs <m$.  
\end{lem}

Applying these two lemmas to Theorem~\ref{thm:main_reconstruction}
we get
\begin{thm}
For a random string $\strs$ over an alphabet of size $\ge4$ the
backtracking algorithm with probability $>1-\delta$ outputs
$\cnfset_\strs$ in time $O_{\delta}\left(n^{1.23\log_{1/p_k} 2}n^2\log n\right)$.
\end{thm}
Recall that $p_k$ falls as $1/k^2$, and with growing $k$,
the algorithm reconstructs entire $\cnfset_\strs$ in 
near quadratic time. 
We now consider strings over alphabet sizes 3 and 2.
\section{Alphabet size 3 and 2}
Recall that 
the harmonic sum $H_n = 1+1/2+\ldots +1/n$ converges to 
$\ln n+\gamma$, where $\gamma$ is the Euler-Mascheroni
constant. 
Therefore, by Lemma~\ref{lem:rnd} we can bound
\[
\EE[\numwt_\strs] < 0.84 \ln n +0.84\gamma.
\]
and Markov's inequality
with probability $\ge 1-\delta$, a random string over alphabet
of size 3 satisfies $\numwt_\strs<\frac{0.84\ln
  n+0.84\gamma}{\delta}$.
Again applying Theorem~\ref{thm:main_reconstruction},
\begin{thm}
The backtracking algorithm for random strings over alphabet size 
3 outputs a subset of $\cnfset_\strs$ containing $\strs$ in time
$O\left(n^{\frac{0.6}{\delta}}\right).$
\end{thm}

For alphabet size $2$,
Lemma~\ref{lem:rnd} shows that 
\[
\EE[\numwt_\strs] < 1.9\sqrt{ n} +0.84\gamma.
\]
Using Markov's inequality and Theorem~\ref{thm:main_reconstruction}
as done for alphabet size 3 yields
\begin{thm}
The backtracking algorithm for random strings over alphabet size 
2 outputs a subset of $\cnfset_\strs$ containing $\strs$ in time
$O\left(2^{\frac{2\sqrt{n}}{\delta}}\right).$
\end{thm}

\section{Conclusions and extensions}
\label{sec:concl-extens}
Starting with protein mass-spectrometry reconstruction, we made two simplifying
assumptions: that all bond cuts are equally likely, and that substring weights imply their compositions.
These two assumptions reduced protein-reconstruction to the simple problem of reconstructing a
string from its substring compositions.
We noted that this is the only unstudied variation of four related reconstruction problems, 
that solving it for binary strings suffices, 
and that it is also a combinatorial simplification of the long-open
turnpike problem.

We called strings with the same  composition multiset \emph{equicomposable},
strings equicomposable only with themselves and their reversal
\emph{reconstructable}, and those with more than these
two trivial equicomposable strings, \emph{confusable}.
We noted that all strings of length at most seven are
reconstructable, interleaved strings and their reversals to obtain
confusable 8-bit strings, and obtained similar 
confusable strings of all lengths one short of a product of two integers $\ge 3$.

Extending polynomial techniques used for the turnpike problem,
we represented strings as bivariate \emph{generating} polynomials.
We used this formulation to characterize equicomposability
in terms of both polynomial multiplication and polynomial factorization,
showing in particular that equicomposable strings are determined
exactly by the prime factorization of their generating polynomials.
We then showed that all strings of lengths not included in the earlier construction,
namely, seven, and one less than either a prime or twice a prime, are
reconstructable.

Interleaving multiple strings we constructed sets of 
$(n+1)^{\log_3 2}$ mutually-equicomposable $n$-bit strings,
and showed a pair of non-interleaved confusable strings.
Using cyclotomic polynomials we upper bounded
the largest number of confusable strings, showing in particular
that the lower bound is tight when the sequence length is one short of
a prime power and twice a prime power.

Many questions remain.
All confusable strings we are aware of are 
described by Theorem~\ref{thm:maxweknow}.
Finding other confusable strings, or proving that
this describes all confusable strings would be of interest.
We made two assumptions. The first implied that we are given 
the composition of all substrings.
What happens when we are given a fraction of all substrings, or of all
substrings of a given length?
The proofs provided here use algebraic arguments.
Direct combinatorial proofs would be of interest.
The second assumption was that all compositions are given
correctly. It would be interesting to know whether some errors can be tolerated.

While prime-related reconstructability may be interesting, 
reconstructability for arbitrary lengths may be of more practical relevance.
It would be interesting to determine whether the lower bound of the
size of equicomposable sets in Theorem~\ref{theorem:lwr_bnd_cnf} is always tight.
It would mean that every string is confusable with at most a sublinear 
number of strings.
A related question is whether most strings of a given length
are reconstructable. This question is related
to the open problem of whether most 0-1 polynomials 
are irreducible over the integers~\cite{OP93}.
A related question addresses the number of composition multisets.
If this number is close to $2^n$ then most strings can be reconstructed.
Another variation is when instead of a string, the bits are arranged
on a ring. 
The constructions presented here extend to ring.
Proving the upper bounds is still open.

Other problems relate to algorithms for reconstructing a string from its substring composition. 
As noted earlier, $n$-bit reconstruction can be reduced to solving a
turnpike problem with $n+1$ exits, and total length $\le n^2$.
This implies a polynomial-time algorithm for the reconstruction.
However such a generic algorithm may have high complexity
and an algorithm that uses the structure of the reconstruction 
problem is of interest.

\section*{Acknowledgements}
We thank Sampath Kannan, Ananda Theertha Suresh and Alex Vardy for
helpful discussions and suggestions.
\bibliographystyle{IEEEtran}

\bibliography{isit10.bib}

\appendix

\section{Proof of Theorem~\ref{thm:lng7}}
\label{sec:appendixB}
The proof follows the proof of Theorem~\ref{thm2p}.
We show that in any factorization 
\[
P_\strs(x,y)=f(x,y)g(x,y),
\]
at least one of $f(x,y)$ and $g(x,y)$ is palindromic.
By Lemma~\ref{lemma:cnf_non_pln_fct}, $\strs$ is reconstructable. 

\begin{align*}
P_\strs(x,x)
&=
1+x+x^2+\dotsb+x^{7}\\
&=
(1+x)(1+x^2)(1+x^4).
\end{align*}
Hence there are only three factorizations of $P_\strs(x,y)$ into two factors. \\[3mm]
\textit{Case 1:}
$f(x,x)=1+x$ and $g(x,x)=(1+x^2)(1+x^4)$,\\
This is identical to \textit{Case 1} of Theorem~\ref{thm2p}. \\[3mm]
\textit{Case 2:}
$f(x,x)=1+x^2$ and $g(x,x)=(1+x)(1+x^4)$,\\
For case 2 we note that $f(x,y)$ has no linear
terms, otherwise it must have both the 
terms $x$ and $y$ which 
would imply the existence of terms of the
form $x^a$ and $y^b$ for positive integers
$a$ and $b$ which violates \textbf{G3}.\\[3mm]
\textit{Case 3:}
$f(x,x)=1+x^4$ and $g(x,x)=(1+x)(1+x^2)$. \\[3mm]
In Case 3 concentrate on $g(x,x)=(1+x)(1+x^2)=1+x+x^2+x^3$.
Just as in the case 2 of the proof of the Theorem~\ref{thm2p} we
can show here that the string is of the form $tut$ where $t$ is a
length 3 string. But this implies that $P_s=P_t(1+x^ay^b)$ where 
$a$ abd $b$ are the number of 0's and 1's in $tu$, so $a+b=4$.
Since $1+x^4$ is irreducible, this factor is palindromic
which proves the third case as well.

\section{Proof of $\subseteq$ in Theorem~\ref{theorem:cnf_of_intr}}
\label{sec:appendixA}
We first show the following lemma.
\begin{lem}
\label{lem:isgen}
Let $P(x,y)$ be a generating polynomial. Any $Q(x,y)\in\mathbb{Z}[x,y]$ with constant term 1
satisfying
\[
P(x^{a+1}y^b,x^ay^{b+1})\rP(x^{a+1}y^b,x^ay^{b+1})
= Q(x,y)\rQ(x,y)
\]
has the form $R(x^{a+1}y^b,x^ay^{b+1})$, where $R(x,y)$ is a generating polynomial.
\end{lem}

\Proof
Similar to the proof of Lemma~\ref{lemma:q_is_gen2},
we can show that 
\begin{equation}
P(x^{a+1}y^b,x^ay^{b+1}) P\left(\frac1{x^{a+1}y^b},\frac1{x^ay^{b+1}}\right)
= Q(x,y)Q\left(\frac1x,\frac1y\right),\label{eq:leftright}
\end{equation}
and $Q(x,y)$ is a 0-1 polynomial.

Note that, all terms in the expansion of the left-hand side of
Equation~\eqref{eq:leftright} has the form
\[
\left(x^{a+1}y^{b}\right)^s
\left(x^{a}y^{b+1}\right)^t
\cdot \frac{1}{\left(x^{a+1}y^{b}\right)^{s'}
\left(x^{a}y^{b+1}\right)^{t'}}
=x^{(a+1)(s-s') +a(t-t')}\,y^{b(s-s')+(b+1)(t-t')}
.\]
For any term $x^iy^j$ present in $Q(x,y)$, consider terms of the form
$x^{i}y^{j}\cdot 1$ in the expansion of the right-hand side of Equation~\eqref{eq:leftright}. Then, for some $s,t$ and $s',t'$,
\begin{align*}
i&=(a+1)(s-s')+a(t-t'),\\
j&=b(s-s')+(b+1)(t-t').
\end{align*}
For simplicity, let $u_{ij}=s-s'\in\mathbb{Z}$ and $v_{i,j}=t-t'\in\mathbb{Z}$. Then
\[Q(x,y)=\sum_{i,j}x^iy^j
=\sum_{i,j}
\left(x^{a+1}y^b\right)^{u_{ij}}\left(x^ay^{b+1}\right)^{v_{ij}}.
\]
Let
$R(x,y)=\sum_{i,j} x^{u_{ij}}y^{v_{ij}}.$
Then $Q(x,y)=R(x^{a+1}y^b,x^ay^{b+1})$.

It remains to show that $R(x,y)$ is a generating polynomial.
Let
$u\df \min\,u_{ij},\text{ and }
v\df \min\,v_{ij}.$
Then
$T(x,y)\df x^{-u}y^{-v}R(x,y)\in\mathbb{Z}[x,y]$.
It's straightforward to see that
\[P(x,y)P\left(\frac{1}{x},\frac{1}{y}\right)
=T(x,y)T\left(\frac{1}{x},\frac{1}{y}\right).
\]
It follows from Lemma~\ref{lemma:q_is_gen2} that $T(x,y)$ is
a generating polynomial. Note that
\[Q(x,y)=
(x^{a+1}y^b)^{u}(x^ay^{b+1})^{v}
\,T\left(x^{a+1}y^b,x^ay^{b+1}\right).\]
Since both $Q(x,y)$ and $T\left(x^{a+1}y^b,x^ay^{b+1}\right)$
are polynomials with constant term $1$, we must have
$(a+1)u+av=0,\mbox{ and }
bu + (b+1)v=0,$
which imply that $u=v=0$. Hence $R(x,y)=T(x,y)$ is a generating function,
and
\[Q(x,y)
=R(x^{a+1}y^b,x^ay^{b+1}).\eqed
\]


\prfarg{Theorem~\ref{theorem:cnf_of_intr}}
It suffices to consider $k=2$, i.e., $\strs=\strs_1\circ \strs_2$.
For $i=1,2$, let $P_i(x,y)$ be the generating polynomial of $S_i$.
By Lemma~\ref{lem:polycomp}, the generating polynomial of $\strs$ is
$P_\strs(x,y)=P_{\strs_1}(x,y)P_{\strs_2}\left(x^{a+1}y^b,x^ay^{b+1}\right),$
where $a$ and $b$ are the numbers of ones and zeroes in $\strs_1$.
Let
\begin{align*}
P_{\strs_1}(x,y)=\prod_{i=1}^{k_1} A_i(x,y),\text{ and }
P_{\strs_2}\left(x^{a+1}y^b,x^ay^{b+1}\right)=\prod_{i=1}^{k_2} B_i(x,y),
\end{align*}
where $A_1,A_2,\dotsc A_{k_1}$  and
$B_1,B_2,\dotsc, B_{k_2}$ are irreducible factors. Then
\[
P_\strs(x,y)=\prod_{i=1}^{k_1} A_i(x,y)\prod_{i=1}^{k_2} B_i(x,y).
\]
Since $\strs\sim \strt$, by Theorem~\ref{thm:plycmp}, there exist $K_1\subseteq [k_1]$ and
$K_2\subseteq [k_2]$ such that
\[
P_{\strt}(x,y)
= \prod_{i\in K_1}A_i(x,y)\prod_{i\in [k_1]\setminus K_1} \rA_i(x,y)
\prod_{i\in K_2}B_i(x,y)\prod_{i\in [k_2]\setminus K_2} \rB_i(x,y).
\]
Let
\begin{align*}
Q_1(x,y)&\df\prod_{i\in K_1}A_i(x,y)\prod_{i\in [k_1]\setminus K_1}
\rA_i(x,y),\\
Q_2(x,y)&\df\prod_{i\in K_2}B_i(x,y)\prod_{i\in [k_2]\setminus K_2} \rB_i(x,y).
\end{align*}
Note that $Q_1(x,y)$ and $Q_2(x,y)$ satisfies
\begin{align*}
Q_1(x,y)\rQ_1(x,y)&= P_{\strs_1}(x,y)\rP_{\strs_1}(x,y),\\
Q_2(x,y)\rQ_2(x,y)&=P_{\strs_2}(x^{a+1}y^b,x^ay^{b+1})\rP_{\strs_2}(x^{a+1}y^b,x^ay^{b+1}).
\end{align*}
The first equation and Lemma~\ref{lemma:q_is_gen2} imply that
$Q_1(x,y)$ is a generating polynomial. Since both $P_{\strt}$ and
$Q_1$ have constant term 1, $Q_2(x,y)$ also has constant term 1.
Then the second equation and Lemma~\ref{lem:isgen} imply that
\[Q_2(x,y)= R(x^{a+1}y^b,x^ay^{b+1}),\]
where $R(x,y)$ is a generating polynomial. By Lemma~\ref{lem:polycomp},
$\strt =\strt_1\circ \strt_2$,
where $\strt_1$ has generating polynomial $Q_1(x,y)$, 
and $\strt_2$ has generating polynomial $R(x,y)$.$\hfill\qed$


\end{document}